\documentclass[lettersize,journal]{IEEEtran}
\usepackage{color,array,amsthm}
\usepackage[ruled,vlined]{algorithm2e}
\usepackage{amssymb}
\usepackage{amsthm}
\usepackage{amsmath}
\usepackage{amsfonts}
\usepackage{yhmath}
\usepackage{amsbsy}
\usepackage{booktabs}
\usepackage{float}
\usepackage{xcolor}
\usepackage{graphicx} 	% 插入latex子图的两个宏包
\usepackage{subfigure}
\usepackage{stfloats}
\usepackage{epstopdf}  % 在导言区添加

\usepackage{lineno,hyperref}
\modulolinenumbers[5]
\usepackage{url}
\usepackage[numbers,sort&compress]{natbib}		% 参考文献连续引用

\usepackage{bm}

\usepackage{xcolor}
\usepackage{xpatch}
\usepackage{makecell}

\theoremstyle{remark}
\newtheorem{theorem}{\hspace{1em}\textit{Theorem}}
\newtheorem{lemma}{\hspace{1em}\textit{Lemma}}

\newtheorem{remark}{\hspace{1em}Remark}

\begin{document}

%\title{Hierarchical Model-Reference Reinforcement Learning Control of Space Triangle Tethered Formation System with Stability Guarantees}
\title{Safe Reinforcement Learning Beyond Baseline Control: A Hierarchical Framework for Space Triangle Tethered Formation System}

\author{Xinyi Tao,~Panfeng Huang,~Fan Zhang
        % <-this % stops a space
\thanks{This work was supported in part by the National Natural Science Foundation of China under Grant Nos.62222313, 62327809, 62573356.\textit{(Corresponding author: Fan Zhang.)}}% <-this % stops a space
\thanks{The authors are with the Research Center for Intelligent Robotics, School of Astronautics, Northwestern Polytechnical University, Xi'an, 710072 China, E-mail: (\href{xytao@mail.nwpu.edu.cn}{xytao@mail.nwpu.edu.cn};   \href{pfhuang@nwpu.edu.cn}{pfhuang@nwpu.edu.cn}; \href{fzhang@nwpu.edu.cn}{fzhang@nwpu.edu.cn})}}

% The paper headers\comment{text}{comment}

%{Tao\MakeLowercase{\textit{et al.}}: Hierarchical Model-Reference Reinforcement Learning Control of TTFS}

%\IEEEpubid{0000--0000/00\$00.00~\copyright~2025 IEEE}
% Remember, if you use this you must call \IEEEpubidadjcol in the second
% column for its text to clear the IEEEpubid mark.

\maketitle

\begin{abstract}
 Triangular tethered formation system (TTFS) provide a promising platform for deep space exploration and distributed sensing due to its intrinsic spatial-orientation stability and capability of adjusting distances among node satellites through deployment and retrieval of tethers. 
% However, the coupled tether–satellite dynamics, actuator constraints, and disturbance sensitivity pose significant challenges to deployment control. 
 However, due to the coupled tether–satellite dynamics and disturbance sensitivity of TTFS, traditional control methods struggle to achieve a balanced trade-off among configuration accuracy requirements, tension constraints, and energy efficiency consumption throughout the deployment process.
 In this paper, a novel model-reference reinforcement learning control framework is proposed for TTFS. By integrating baseline model-based control with a Soft Actor–Critic (SAC) compensator, the proposed method simultaneously achieves high-precision tracking, fuel efficiency, and compliance with tension limits. A hierarchical training scheme is developed to address the convergence difficulties arising from strongly coupled states in centralized training, while tailored reward functions, reset conditions, and normalization criteria are designed to accelerate training convergence. Closed-loop stability of the overall control law is rigorously proven using Lyapunov methods. Simulation results demonstrate that the proposed controller reduces steady-state tracking errors by over 96\% for tethers and 99\% for node satellites, while cutting fuel consumption by two orders of magnitude compared with the baseline method. These results validate the effectiveness and stability of the proposed approach for TTFS deployment control.
% Triangular tethered formation system (TTFS) provide a promising platform for deep space exploration and distributed sensing due to its intrinsic spatial-orientation stability and capability of adjusting distances among node satellites through deployment and retrieval of tethers. However, the coupled tether–satellite dynamics, actuator constraints, and disturbance sensitivity pose significant challenges to deployment control. In this paper, a novel model-reference reinforcement learning control framework is proposed for TTFS. By integrating baseline model-based control with a Soft Actor–Critic (SAC) compensator, the proposed method simultaneously achieves high-precision tracking, fuel efficiency, and compliance with tension limits. A hierarchical training scheme is developed to address the convergence difficulties arising from strongly coupled states in centralized training, while tailored reward functions, reset conditions, and normalization criteria are designed to accelerate training convergence. Closed-loop stability of the overall control law is rigorously proven using Lyapunov methods. Simulation results demonstrate that the proposed controller reduces steady-state tracking errors by over 96\% for tethers and 99\% for node satellites, while cutting fuel consumption by two orders of magnitude compared with the baseline method. These results validate the effectiveness and stability of the proposed approach for TTFS deployment control.
\end{abstract}


\section{INTRODUCTION}
\IEEEPARstart{T}riangular tethered formation system (TTFS) is a kind of flexible and controllable large-scale coupling system \cite{hong2025review} with great potential in deep space exploration, large baseline interferometry, and distributed sensing, etc \cite{chung2005dynamics,moccia1993attitude,farley2004tethered}. Compared with free-flying formation systems, TTFS has longer life cycle for its spatial-orientation stability in spinning case \cite{zhang2021stable} and allows variable separations among nodes through deployment and retrieval of tethers. Howerver, tether-satellite coupling, disturbance sensitivity, and actuation limits during long-term in-orbit operation challenge the control strategy design. Deployment control of TTFS is a critical issue must be addressed \cite{kumar2006review}. 

Deployment control approaches for TTFS are categorized into PID control \cite{djebli1999laws, chen2020analysis,matunaga2001coordinated}, sliding mode control(SMC) \cite{9743568,li2022fractional,zhang2023fractional,huang2023event,2022huangevent,huang2023event1,mohsenipour2020robust} and optimization-based control \cite{rigatos2024nonlinear,guang2019optimal,li2020discrete}. The first two control approaches are primarily concerned with 
trajectory tracking accuracy, but they frequently overlook actuator limits and energy comsumption. Under coupling dynamics and disturbances those methods are effective, but may result in large energy comsumption and careful tuning to avoid tension violations. For instance, \cite{sun2025adaptive} develops a fault-tolerant sliding mode control strategy which performs well in trajectory tracking, but the control input exceeds 50N. After converting the generalized control input in \cite{huang2023event,2022huangevent,huang2023event1} to actual value, it was discovered exceeding the actuator limit. Optimization-based control methods include traditional optimal control \cite{kang2019unified}, model predictive control(MPC) \cite{li2020discrete}. These methods depend on high real-time computing complexity and precise TTFS modeling. 

Recently, reinforcement learning (RL) as another optimization-based method has been applied in diverse control settings, including maritime vessel path following under wind and wave disturbances \cite{zheng2022soft}, autonomous racing \cite{song2023reaching}, disturbance‐mitigation in cyber-physical systems \cite{wu2023deep} and helicopter system \cite{11084854},etc. Among RL algorithms, Soft Actor-Critic (SAC) stands out due to its desirable properties \cite{haarnoja2018soft1}: entropy regularization that encourages exploration while balancing exploitation; off-policy learning enabling sample reuse; robustness to hyperparameter variations and empirical stability in continuous control tasks. In contrast, on-policy methods like PPO \cite{schulman2017proximal}, TRPO \cite{schulman2015trust} often require more data and are more prone to getting trapped in local optima, whereas purely value-based methods like DQN \cite{mnih2015human} may struggle in continuous action spaces.

However, applying RL to solve the deployment control problem of TTFS still poses three principal challenges. Firstly, it is challenging for model-free RL
to ensure closed-loop stability \cite{2021Model}. Secondly, high-dimensional and strongly coupled states might impede RL training convergence, resulting in slower learning and increased risk of poor local minima \cite{tipaldi2022reinforcement,dong2024reinforcement}. Thirdly, combined with the approximation errors, limited exploration, and potential instability during policy learning, RL tends to provide conservative adjustments rather than disruptive improvements, making it challenging to substantially outperform the baseline control law \cite{tang2025deep}.

With the consideration of the merits and limitations of
existing RL methods, in this paper, \IEEEpubidadjcol we propose a novel learning-based control strategy for TTFS that combines a baseline controller and a learned SAC compensator. The control objective is to optimize a composite performance measure under external disturbances that takes tracking precision, fuel consumption, and tether tension limits into account. The major contributions are as follows:
\begin{enumerate}
	%	\item Build the complete dynamics model of TTFS consisting of satellites, tethers and winches for tension calculation.
	%	\item Design fixed-time consensus protocols to ensure the acccuracy of formation tracking and the consistency between tethers and connected satellites.
	%	\item Propose a new decentralized control framework which shatters the assumption that tethers are in tension during the whole deployment of TTFS to reflect the actual relationship between tethers and satellites. Meanwhile, this framework lightens communication and computation burden of the central node.
	\item Propose a model-referenced RL control scheme that integrates the advantages of model-based control and model-free reinforcement learning, accompanied by a theoretical proof of closed-loop system stability.
	\item The proposed model-reference RL algorithm eliminates the need for specific values of uncertainty characteristics, such as Lipschitz constants, bounds, and frequencies.
	\item Propose a novel hierarchical training framework that resolves the convergence difficulties in centralized training framework caused by tether-spacecraft coupling dynamics. 
%	\item Design a set of tailored environmental reset conditions, reward functions, and normalization criteria to accelerate training convergence.
	\item With tailored environmental reset conditions, reward functions, and normalization criteria, the training convergence is accelerated and the propsed algorithm is significantly superior to the baseline algorithm in terms of both tracking accuracy and energy consumption.
	%	\item Derive the settling time that does not depend on the initial condition of the system based on fixed-time control theory to gurantee finite-time convergence of system deployment.
\end{enumerate}

The remaining paper is arranged as follows:
In Section \ref{Sec:sec2} the TTFS model is established..
Section \ref{Sec:sec3} gives the control framework of the model-reference rl control. Precise implementation details are introduced in Section \ref{Sec:sec4} Section \ref{Sec:sec5} demonstrates the closed-loop stability proof. Numerical simulation is carried out in Section \ref{Sec:sec6} to illustrate the effectiveness of the controller. Finally, Section \ref{Sec:sec7} summarizes the paper.

\textbf{Notations:}$\left(\cdot\right)_x$and$\left(\cdot\right)_y$refer to  the $x$-axis and $y$-axis component of a vector.$\|\cdot\|_2$denotes the Euclidean norm. $\mathbb{R}$ and $\mathbb{R}^n$ are sets of real numbers and n-dimensional real vectors respectively. $\mathbb{E}_(\cdot)\left[\cdot\right]$ represents the expectation operator. $\operatorname{anti\text{-}diag}\{a_1,a_2,\cdots,a_n\}$ denotes an anti-diagonal matrix whose non-zero entries lie on the anti-diagonal, from the top-right to the bottom-left

\section{SYSTEM DYNAMICS}
\label{Sec:sec2}
TTFS is a rigid-flexible system consisting of 3 satellites connected by $3$ tethers which are driven by reel motors.

%The system consists of 3 satellites connected by $3$ tethers. As is shown in Fig. \ref{fig:Fig1}, the whole system forms a close-loop triangle rotating around the system centroid that rotates around the earth. During the process of development, tethers are driven by satellite-mounted actuated winches. In order to avoid collision, tethers are supposed to remain micro-tensioned or micro-relaxed.
As is shown in Fig. \ref{fig:Fig1}, $I-XYZ$ is the inertial frame where the origin $I$ located at the center of the Earth. Axis $IX$ is directed toward the ascending node of the orbit; axis $IZ$ is normal to the orbital plane, and the $IY$-axis is defined to complete a right-handed coordinate system.  The orbit frame $o-xyz$ is attached to the centroid of the TTFS. Axis $ox$ is aligned opposite to the velocity vector of the system centroid, the $oy$-axis points points radially outward from the Earth to the system centroid, while the $oz$-axis is determined by the right-hand rule. To simplify the analysis process, the following assumptions are usually made:

%Define $I-XYZ$ as the inertial frame with its origin $I$ located at the center of the Earth, in which the $IX$-axis points to the ascending intersection, the $IZ$-axis is perpendicular to the orbital plane, and the $IY$-axis completes the right-handed triad.

%Define $o-xyz$ as the orbit frame, where the origin $o$ locates in the center of mass of the formation system, the $ox$-axis points to the opposite direction of the orbital motion of the system centroid, the $oy$-axis points from the center of earth to the center system centroid, while the $oz$-axis completes a right-hand coordinate frame. In this paper, the system model can be further simplified according to the following assumptions.
\begin{enumerate}
	\item The motion of the formation is restricted within the orbital plane, and any out-of-plane dynamics are ignored.
	\item The centroid of system follows a circular Keplerian trajectory, and its orbital angular velocity is considered constant and denoted by $\omega_0$.
	\item Each satellite in TTFS is modeled as a point mass, as its physical dimensions are negligible relative to the characteristic length of the formation.
	\item The tethers are modeled as elastic but incompressible rods. This assumption is justified by the ability of the reel mechanisms to maintain the tethers a microtensioned or microrelaxed state during deployment.
	%	\item The obit radius $R$ of the system centroid is much larger than the distance between satellites and the system centroid $r$, namely $\frac{r}{R}\ll1$
\end{enumerate}\vspace*{-.5pc}
\begin{figure}[H] 
	\centerline{\includegraphics[width=0.83\linewidth]{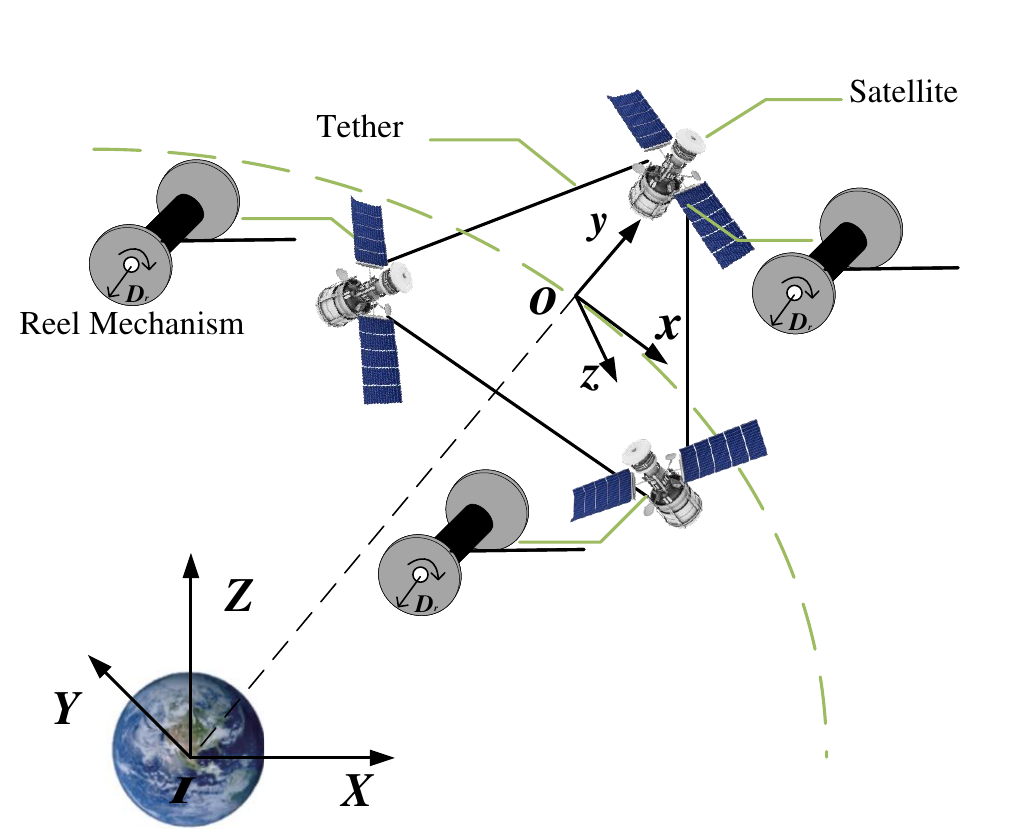}}
	\caption{Schematic diagram of space triangle tethered formation system}
	\label{fig:Fig1}
\end{figure}
Define $\boldsymbol{r}_{i}=\left[x_i,y_i\right]^T\in\mathbb{R}^2$as the position vector of the $i$-th satellite in the orbit frame, $\boldsymbol{T}_{i}$ as tension acting on the $i$-th satellite, $\boldsymbol\tau_i $ as the acceleration provided by thrust, $\boldsymbol{d}_i$ as the space perturbations. The dynamic model of satellites in the orbit frame can be expressed as follows:
\begin{equation}
	\label{E1}
	\left\{
	\begin{aligned}
		&\ddot{x}_i=&&\!\!\!\!\!2ny_i+n^2x-\frac{\mu x_i}{[x_i^2+(R+y_i)^2+z_i^2]^{\frac{3}{2}}}\\
		&&&\!\!\!\!\!+\left(\frac{\boldsymbol{T}_{i}}{m_i}\right)_x+\left(\boldsymbol{\tau}_i\right)_x+(\boldsymbol{d}_i)_x\\
		&\ddot{y}_i=&&\!\!\!\!\!-\!2n\dot{x}_i\!+\!n^2(R+y_i)\!-\!\frac{\mu(y_i\!+\!R)}{[x_i^2\!+\!(R\!+\!y_i)^2\!+\!z_i^2]^{\frac{3}{2}}}\\
		&&&\!\!\!\!\!+(\frac{\boldsymbol{T}_{i}}{m_i})_y+\left(\boldsymbol{\tau}_i\right)_y+(\boldsymbol{d}_i)_y
	\end{aligned}
	\right.
\end{equation}
where $n$ represents the orbital angular velocity, and $i,j,k = 1,2,3$.

The tension force is modeled as:
\begin{equation}
	\label{E2}
	\begin{aligned}
		\boldsymbol{T}_{i} = 
		& -\frac{K_T}{l_{ij}}\left(|\boldsymbol{r}_i - \boldsymbol{r}_j| - l_{ij} \right)
		\frac{\boldsymbol{r}_i - \boldsymbol{r}_j}{|\boldsymbol{r}_i - \boldsymbol{r}_j|} \\
		& -\frac{K_T}{l_{ik}}\left(|\boldsymbol{r}_i - \boldsymbol{r}_k| - l_{ik} \right)
		\frac{\boldsymbol{r}_i - \boldsymbol{r}_k}{|\boldsymbol{r}_i - \boldsymbol{r}_k|} \\
		K_T = 
		& \begin{cases}
			0,  & |\boldsymbol{r}_i - \boldsymbol{r}_j| \leq l_{ij} \\
			EA, & |\boldsymbol{r}_i - \boldsymbol{r}_j| > l_{ij}
		\end{cases}
	\end{aligned}
\end{equation}

where $K_T$, $E$ and $A$ represent the elasticity coefficient, the Young's modulus and the cross-sectional area of tethers respectively. Denote $l_{12}=l_{21}=l_1$, $l_{23}=l_{32}=l_2$, $l_{13}=l_{31}=l_3$.

The dynamic model of reel motors can be expressed as:
\begin{equation}
	\label{E3}
	\left\{\begin{aligned}
		&\frac{T_m}{D_r}\ddot{l}_i+\frac{1}{D_r}\dot{l}_i+\delta_i(l_i,\dot{l_i})=\frac{1}{k_e}\nu_i\\
		&T_m=\frac{R_aJ}{k_ek_m}
	\end{aligned}
	\right.
\end{equation}
where $l$ is the natural (unstressed) length of the tether; $\nu_i$ is the voltage applied to the armature; $\delta(l_i,\dot{l}_i)$ represents unmodeled
dynamics; $T_m$,$R_a$,$k_e$, $k_m$ and $J$ are parameters of reel motors. Detailed explanations are given in our previous work \cite{tao2023fixed}.

\section{Control Scheme Design}
\label{Sec:sec3}
\subsection{Preliminaries}
Denote $\boldsymbol{v}_i=\left[v_{ix},v_{iy}\right]^T\in\mathbb{R}^2$ as the velocity vector of the $i$-th satellite in the orbit frame. Let $\boldsymbol{r}_{sat}=\left[\boldsymbol{r}_1^T,\boldsymbol{r}_2^T,\boldsymbol{r}_3^T\right]^T$, $\boldsymbol{v}_{sat}=\left[\boldsymbol{v}_1^T,\boldsymbol{v}_2^T,\boldsymbol{v}_3^T\right]^T$, $\boldsymbol{\varepsilon}=\left[\boldsymbol{r}_{sat}^T,\boldsymbol{v}_{sat}^T\right]^T$, $\boldsymbol{u}=\left[\boldsymbol{\tau}_1^T,\boldsymbol{\tau}_2^T,\boldsymbol{\tau}_3^T\right]^T$, $\boldsymbol{T}=\left[\boldsymbol{T}_1^T,\boldsymbol{T}_2^T,\boldsymbol{T}_3^T\right]^T$, $\boldsymbol{D}=[\boldsymbol{d}_1^T,\boldsymbol{d}_2^T,\boldsymbol{d}_3^T]^T$. Eq. \eqref{E1} can be reformulated as: 
\begin{equation}
	\label{E4}
	\setlength{\arraycolsep}{1pt} % 调整矩阵列间距，默认是 5pt
	\dot{\boldsymbol{\varepsilon}} \!\!=\!\!
	\left[
	\begin{array}{cc}
		0 & I_6 \\
		\boldsymbol{A}(\boldsymbol{r}_{sat}) & I_3 \otimes C
	\end{array}
	\right]
	\! \boldsymbol{\varepsilon} \!+\!
	\left[
	\begin{array}{c}
		0 \\
		I_6
	\end{array}
	\right]
	\!\!(\boldsymbol{u}\!+\!\boldsymbol{T}) \!+\!
	\left[
	\begin{array}{c}
		0 \\
		\boldsymbol{G}(\boldsymbol{r}_{sat}) \!+\! \boldsymbol{D}
	\end{array}
	\right]
\end{equation}
where $\boldsymbol{A}(\boldsymbol{r}_{sat})=\operatorname{diag}\{\rho_1,\rho_2,\rho_3\}\bigotimes I_2$, $\rho_i=\left[x_i^2+\right.\\\left.(R+y_i)^2+z_i^2\right]^{-\frac{3}{2}}$, $C=I_3\bigotimes\operatorname{anti\text{-}diag}\{2n,-2n\}$, $\boldsymbol{G}(\boldsymbol{r}_{sat})=\left[\boldsymbol{g}_1^T,\boldsymbol{g}_2^T,\boldsymbol{g}_3^T\right]^T$, $\boldsymbol{g}_i=\left[0,\mu R^{-2}-\mu R\rho_i\right]^T$.

Assume that the accurate model Eq.\eqref{E4} is not available. Neglecting the disturbance term, the nonlinear tension term, and the second-order infinitesimal terms from the Taylor expansion of $\rho_i$, the nominal linear model of Eq.\eqref{E4} is obtained:
\begin{equation}
	\label{E5}
	\setlength{\arraycolsep}{5pt} % 调整矩阵列间距，默认是 5pt
	\dot{\boldsymbol{\varepsilon}}_m =
	\left[
	\begin{array}{cc}
		0 & I_6 \\
		A_m & I_3 \otimes C
	\end{array}
	\right]
	\boldsymbol{\varepsilon}_m+
	\left[
	\begin{array}{c}
		0 \\
		I_6
	\end{array}
	\right]
	\boldsymbol{u}_m
\end{equation} 
where $A_m=\operatorname{diag}\{0,3n^2\}\bigotimes I_3$.

Denote $\boldsymbol{\iota}=\left[l_1,l_2,l_3\right]^T\in\mathbb{R}^3$ as the stacked vector of 3 tethers, $\boldsymbol{\chi}=\left[\dot{l}_1,\dot{l}_2,\dot{l}_3\right]$. Let $\boldsymbol{\eta}=\left[\boldsymbol{\iota}^T,\boldsymbol{\chi}^T\right]^T$, $\boldsymbol{\nu}=\left[\nu_1,\nu_2,\nu_3\right]^T$, $\boldsymbol{\delta}=\left[\delta_1,\delta_2,\delta_3\right]^T$. Eq. \eqref{E3} can be reformulated as:
\begin{equation}
	\label{E6}
	\setlength{\arraycolsep}{2pt} % 调整矩阵列间距，默认是 5pt
	\dot{\boldsymbol{\eta}} =
	\left[
	\begin{array}{cc}
		0 & I_3 \\
		0 & -\frac{1}{T_m}{I}_3 
	\end{array}
	\right]
	\boldsymbol{\eta}+
	\left[
	\begin{array}{c}
		0 \\
		\frac{D_r}{T_mK_e}I_3
	\end{array}
	\right]
	\boldsymbol{\nu}+
	\left[
	\begin{array}{c}
		0 \\
		\frac{D_r}{T_m}I_3
	\end{array}
	\right]\boldsymbol{\delta}
\end{equation}

By neglecting unmodeled dynamics terms, Eq.\eqref{E6} can be rewritten as:
\begin{equation}
	\label{E7}
	\setlength{\arraycolsep}{5pt} % 调整矩阵列间距，默认是 5pt
	\dot{\boldsymbol{\eta}}_m =
	\left[
	\begin{array}{cc}
		0 & I_3 \\
		0 & -\frac{1}{T_m}I_3 
	\end{array}
	\right]
	\boldsymbol{\eta}_m+
	\left[
	\begin{array}{c}
		0 \\
		\frac{D_r}{T_mK_e}I_3
	\end{array}
	\right]
	\boldsymbol{\nu}_m
\end{equation} 

Suppose there exist control laws $\boldsymbol{u}_m$ and $\boldsymbol{\nu}_m$ such that the states of the nominal systems \eqref{E5} and \eqref{E6} track the reference signals $\boldsymbol{\varepsilon}_r$ and $\boldsymbol{\eta}_m$, i.e., $\|\boldsymbol{\varepsilon}_m-\boldsymbol{\varepsilon}_r\|\rightarrow0$, $\|\boldsymbol{\eta}_m-\boldsymbol{\eta}_r\|\rightarrow0$ as $t\rightarrow\infty$.
\begin{lemma}
	\cite{khalil2002nonlinear}
	consider a nonlinear system:
	\begin{equation}
		\label{E8}
		\dot{\boldsymbol{\xi}}=f(\boldsymbol{\xi},\boldsymbol{\zeta},\boldsymbol{\sigma})
	\end{equation}
	where $\boldsymbol{\sigma}$ is bounded disturbance. A control law $\boldsymbol{\zeta}$ is considered admissible in relation to the system \eqref{E8} if it is capable of making $\boldsymbol{\xi}$ uniformly ultimately bounded, i.e., there exists a continuously differentiable function $V(\boldsymbol{\xi})$ associate with $\boldsymbol{\zeta}$
	\begin{align}
		\psi_1(\|\boldsymbol{\xi}\|_2^2) &\le V(\boldsymbol{\xi}) \le \psi_2(\|\boldsymbol{\xi}\|_2^2), \label{E9} \\
		\dot{V}(\boldsymbol{\xi}) &\le -W(\boldsymbol{\xi}) + \psi_3(\|\boldsymbol{\sigma}\|_2^2), \label{E10} \\
		W(x) &> \psi_3(\|\boldsymbol{\sigma}\|_2^2), \quad \forall \|\boldsymbol{\xi}\|_2^2 > r, \label{E11}
	\end{align}
	for some constant $r > 0$. $\psi_1(\cdot)$, $\psi_2(\cdot)$, $\psi_3(\cdot)$ are class $\mathcal{K}$ functions, $W(\cdot)$
	is a continuous positive definite function.
\end{lemma}
%The control framework is shown in Fig. \ref{fig:Fig11}. There are 2 local controllers for the $i$th satellite. Tether length controller for the motor safeguards consensus for tether length. Satellites' formation controller for the thrust safeguards consensus for satellites' position. The desired tether length is set a little bit shorter than the distance among satellites to ensure that each tether is in a micro-tensioned state. Each member satellite broadcast tether length and position information to neighbors. Communication topology is determined by geometric constraints. 
\subsection{Controller  Structure}
The overall control sturcture is shown in Fig.\ref{fig:Fig11}. $\boldsymbol{u}$ and $\boldsymbol{\nu}$ consist of 2 parts:
\begin{equation}
	\label{E12}
	\begin{aligned}
		&\boldsymbol{u}=\boldsymbol{u}_b+\boldsymbol{u}_l\\
		&\boldsymbol{\nu}=\boldsymbol{\nu}_b+\boldsymbol{\nu}_l
	\end{aligned}
\end{equation}
where $\boldsymbol{u}_b$ and $\boldsymbol{\nu}_b$ are baseline control laws whose form and parameters are identical to those of the nominal system control law $\boldsymbol{u}_m$ and $\boldsymbol{\nu}_m$. $\boldsymbol{u}_b$ and $\boldsymbol{\nu}_b$ are admissible control of system \eqref{E4} and \eqref{E6} respectively.
$\boldsymbol{u}_l$ and $\boldsymbol{\nu}_l$ are control laws generated by neural networks trained through RL. The training process of 
$\boldsymbol{u}_l$ and $\boldsymbol{\nu}_l$ are described in detail in Section \ref{Sec:sec4}.
\begin{figure}[htp] 
	\centerline{\includegraphics[width=0.83        \linewidth]{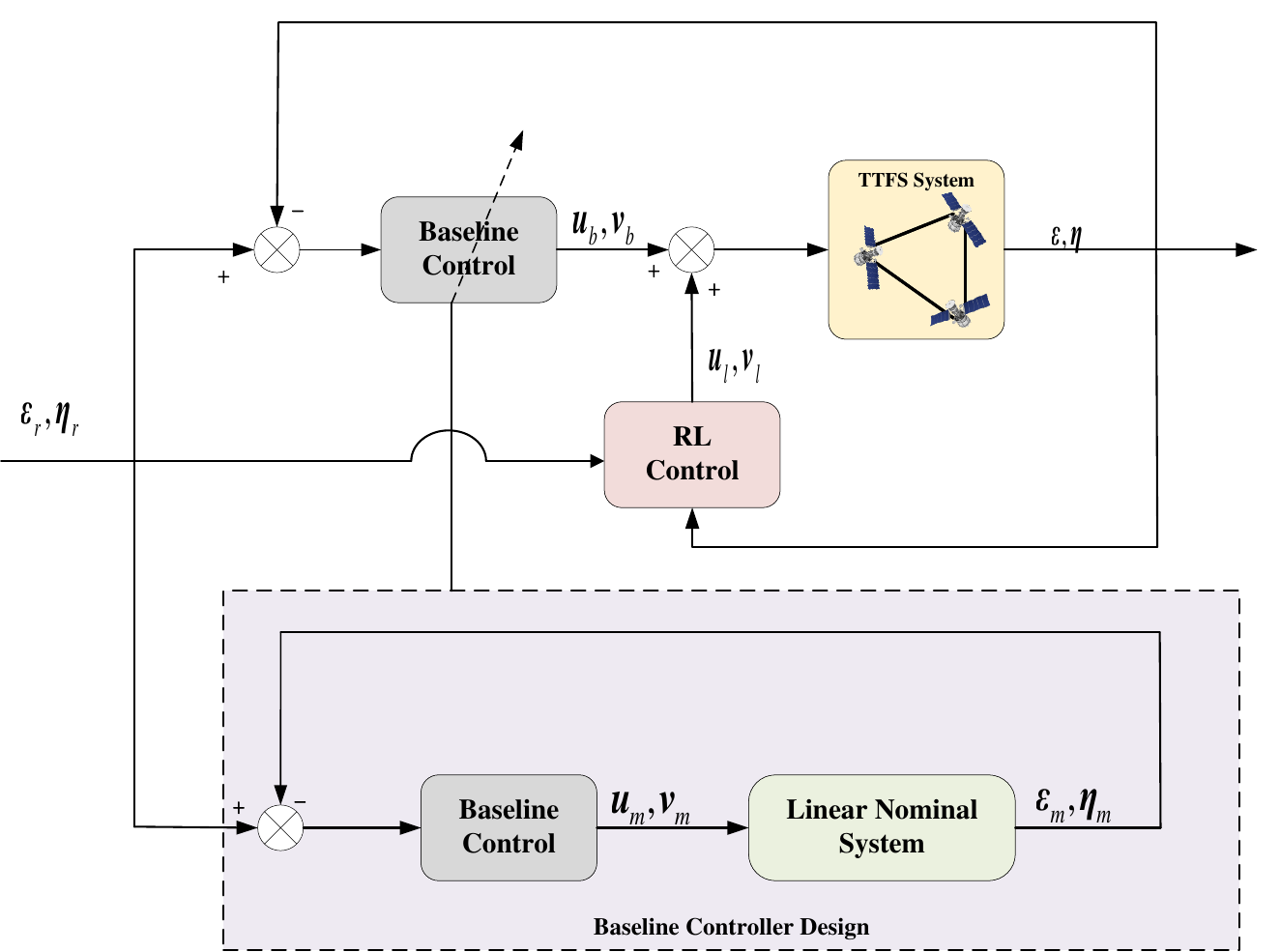}}
	\caption{ Control structure for TTFS}
	\label{fig:Fig11}
\end{figure}
%\begin{assumption}
%	\label{assumption1}
%	The given uncertainty should be Lipschitz with respect to $t$ while satisfying $|d_{i,k}|\leq D_{i,k}$ ,$i\in\{1,2,3\}$,$k\in\{1,2\}$.
%\end{assumption}
\begin{remark}
	\label{remark1}
	Baseline control laws $\boldsymbol{u}_b$ and $\boldsymbol{\nu}_b$ are employed to ensure the basic tracking performance and can be designed  based on nominal models \eqref{E5} and \eqref{E7} using any existing method. In this paper, we choose PD method for baseline control design.
\end{remark}

\section{Design and Implementation of Model-Reference Reinforcement Learning Algorithm}
\label{Sec:sec4}
\subsection{Markov Decision Process}
In the context of RL, dynamics \eqref{E4}, \eqref{E6} are represented by the mathematical model of a Markov decision process, denoted by the tuple $\mathcal{MDP}:=\langle\mathcal{S,U,P},R,\gamma\rangle$, in which $\mathcal{S}$ denotes the state space, $\mathcal{U}$ is the action space, $\mathcal{P:S\times U\times S}\rightarrow\mathbb{R}$ specifies a transition
probability, $R:\mathcal{S\times U}\rightarrow\mathbb{R}$ denotes a reward function, and $\gamma\in\left[0,1\right)$ is a discount factor. For reel motors' environment(referred to as  tethers' environment hereafter) \eqref{E6}, the observable state vector $\boldsymbol{s}_l\in\mathcal{S}$ is chosen as $\boldsymbol{s}_{l-original}=\{\boldsymbol{\eta}_r,\boldsymbol{\eta},\boldsymbol{e}_l\}$, where $\boldsymbol{e}_l=\boldsymbol{\eta}_r-\boldsymbol{\eta}$. For node satellites' environment \eqref{E4}, the observable state vector $\boldsymbol{s}_{sat-original}\in\mathcal{S}$ is chosen as $\boldsymbol{s}_{sat}=\{\boldsymbol{\varepsilon}_r,\boldsymbol{\varepsilon},\boldsymbol{e}_{sat}\}$, where $\boldsymbol{e}_{sat}=\boldsymbol{\varepsilon}_r-\boldsymbol{\varepsilon}$.

To mitigate training instability due to scale differences between position/tether lengths and velocity/tether reeling speed, the original observable vectors $\boldsymbol{s}_{l-original}$ and $\boldsymbol{s}_{sat-original}$ are individually normalized using the "Min-Max" method \cite{cabello2023impact}:
\begin{equation}
	\label{E13}
	\boldsymbol{s}_l=\frac{\boldsymbol{s}_l-\boldsymbol{s}_l^{min}}{\boldsymbol{s}_l^{max}-\boldsymbol{s}_l^{min}},
	\boldsymbol{s}_{sat}=\frac{\boldsymbol{s}_{sat}-\boldsymbol{s}_{sat}^{min}}{\boldsymbol{s}_{sat}^{max}-\boldsymbol{s}_{sat}^{min}}
\end{equation}
Actions $\boldsymbol{u}_{l}^{norm},\boldsymbol{\nu}_{l}^{norm}$ output by neural networks are constrained to 
$\left[0,1\right]$ and are denormalized via the inverse "Min–Max" transformation before being fed into the environment:
\begin{equation}
	\label{E14}
	\begin{aligned}
		&\boldsymbol{u}_l=\boldsymbol{u}_l^{min}+\boldsymbol{u}_{l}^{norm}\left(\boldsymbol{u}_l^{max}-\boldsymbol{u}_l^{min}\right),\\
		&\boldsymbol{\nu}_{l}=\boldsymbol{\nu}_l^{min}\boldsymbol{\nu}_{l}^{norm}\left(\boldsymbol{\nu}_l^{max}-\boldsymbol{\nu}_l^{min}\right)
	\end{aligned}
\end{equation}
where the superscripts "min" and "max" denote the respective minimum and maximum values of each original vector.

Since RL learns control policies from sampled data, assuming that we can sample discrete-time input and state data from systems \eqref{E4} and \eqref{E6}. Denote the observable state vectors at time step $t$ are $\boldsymbol{s}_{l-original,t}=\{\boldsymbol{\eta}_{r,t},\boldsymbol{\eta}_t,\boldsymbol{e}_{l,t}\}$ and $\boldsymbol{s}_{sat-original,t}=\{\boldsymbol{\varepsilon}_{r,t},\boldsymbol{\varepsilon}_t,\boldsymbol{e}_{sat,t}\}$. The actions at time step $t$ are $\boldsymbol{u}_{l,t}$ and $\boldsymbol{\nu}_{l,t}$. Meanwhile the normalized ones are denoted as $\boldsymbol{s}_{sat,t}$ , $\boldsymbol{s}_{l,t}$, $\boldsymbol{u}_{l,t}^{norm}$, $\boldsymbol{\nu}_{l,t}^{norm}$.

\subsection{Hierarchical Training Framework}
The hierarchical training frame is illustrated in Fig. \ref{fig:Fig12}. Firstly,the normalized control policy $\boldsymbol{\nu}_l^{norm}$ is trained until it converges through tethers' environment. The converged control policy is then embedded into the node satellites' environment together with baseline control laws $\boldsymbol{\nu}_b$ to provide real-time control actions for reel motors, after which the node-satellites–level control policy $\boldsymbol{u}_l^{norm}$ is trained until convergence.
\begin{figure}[htp] 
	\centerline{\includegraphics[width=1      \linewidth]{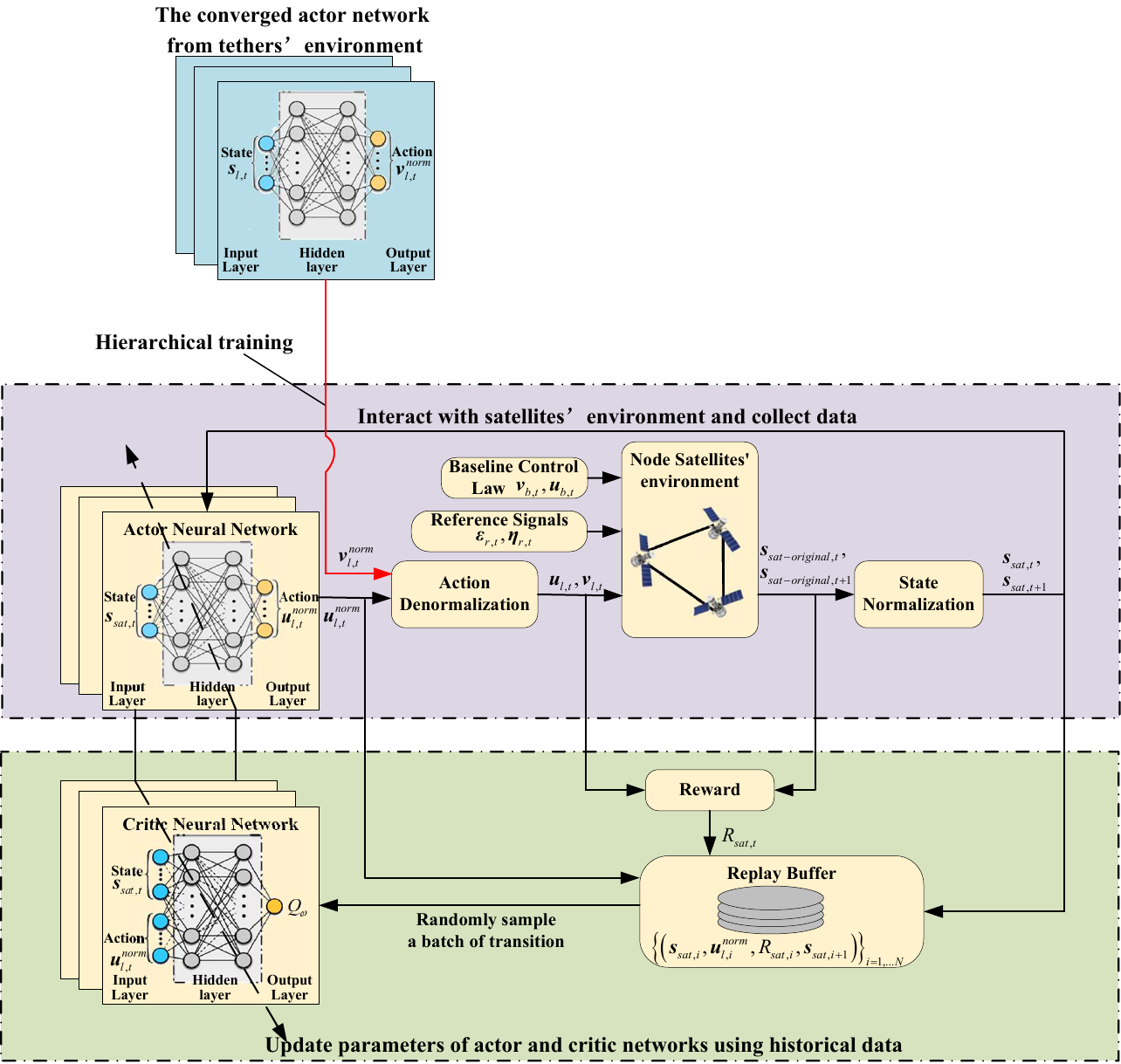}}
	\caption{ Control structure for TTFS}
	\label{fig:Fig12}
\end{figure}

For $\boldsymbol{u}_l$, the control policy of node satellites, the training process consists of two steps: interaction and update. During interaction, firstly, actor networks select an action vector $\boldsymbol{u}_{l,t}^{norm}$ based on current states $\boldsymbol{s}_{sat,t}$. Then agents interact with the environment to obtain rewards $R_{sat,t}$ and states $\boldsymbol{s}_{sat,t+1}$,  at the next timestep $t+1$. The historical data are stored as a transition $\left(\boldsymbol{s}_{sat,t},\boldsymbol{u}_{l,t}^{norm},R_{sat,t},\boldsymbol{s}_{sat,t+1}\right)$ in the replay buffer. In the update process, a batch of transition data is sampled and fed into the critic network. The critic network updates its parameters by minimizing the Bellman error with an entropy term, while the actor network updates its parameters by computing the entropy-regularized action-value function. The interaction and update processes are iterated until  policy converges to the optimal one. The training process of $\boldsymbol{\nu}_l$, the control policy of tether-level follows the same procedure with transition $\left(\boldsymbol{s}_{l,t},\boldsymbol{\nu}_{l,t}^{norm},R_{l,t},\boldsymbol{s}_{l,t+1}\right)$.

\subsection{Environment Specification}
\label{Sec:env}
This section presents the design of the environment, including reward functions, termination conditions, and reset conditions.
\subsubsection{Reward Functions}
\label{Sec:reward}
In RL-based control, dynamic and disturbance constraints are incorporated into the environment model, whereas performance-related constraints are addressed through the design of reward functions. Nonlinear functions (e.g., {\rm Tanh} and power functions) are utilized to constrain the reward range during training, which helps to mitigate the risk of gradient explosion.

According to the configuration tracking constraints of TTFS, the reward function of the tethers' environment is composed of the tether length error and the reeling speed error as follows:
\begin{equation}
	\label{E15}
	R_{l,t}\! =\!\alpha_1\!\!\left(\frac{1}{1\!+\!\|H_1\boldsymbol{e}_{l,t}\|_2}\!-\!1\right)\!+\!\alpha_2{\rm tanh}(-\|H_2\boldsymbol{e}_{l,t}\|_2)
\end{equation}
where $H_1=\left[I_3,0\right]$, $H_2=\left[0,I_3\right]$, $\alpha_1>0$, $\alpha_2>0$.

The reward function of the node satellites' environment comprises three terms: 
\begin{equation}
	\label{E16}
	R_{sat,t}=R_{sat-tracking,t}+R_{sat-energy,t}+R_{sat-dis,t}
\end{equation}
$R_{sat-tracking}$ is derived from position and speed errors, which enforces the system’s configuration tracking constraints:
\begin{equation}
	\label{E17}
	\begin{aligned}
		R_{sat-tracking,t}=&\alpha_3\left(\frac{1}{1+\|H_3\boldsymbol{e}_{sat,t}\|_2}-1\right)\!\\
		+&\alpha_4{\rm tanh}(-\|H_4\boldsymbol{e}_{sat,t}\|_2)
	\end{aligned}
\end{equation}
with $H_3=\left[I_6,0\right]$, $H_4=\left[0,I_6\right]$, $\alpha_3>0$, $\alpha_4>0$.
$R_{sat-energy}$ denotes an energy reward based on thruster control inputs, capturing the system’s energy constraints:
\begin{equation}
	\label{E18}
	\begin{aligned}
		R_{sat-energy,t}=\beta_1\left(\frac{1}{1+\beta_2\|\boldsymbol{u}_t\|_2}-1 \right)
	\end{aligned}
\end{equation}
where $\beta_1>0$, $\beta_2>0$.
The third term $R_{sat-dis,t}$ represents a tension penalty defined by the deviation between the inter-satellite distances and corresponding tether lengths, ensuring tethers maintain slightly tensioned/ slightly relaxed states.
% \begin{equation}
	% 	\label{E19}
	% 	\begin{aligned}
		% 		R_{sat-dis,t}=& \frac{1}{1\!+\gamma_1\left(\frac{l_{1,t}}{\|\boldsymbol{\rho_{1,t}}-\boldsymbol{\rho_{2,t}}\|}_2\!-\!1\right)^2}\!+\!\frac{1}{1\!+\gamma_1\left(\frac{l_{2,t}}{\|\boldsymbol{\rho_{,t}}-\boldsymbol{\rho_{3,t}}\|}_2\!-\!1\right)^2}\\
		% 		+&\frac{1}{1\!+\gamma_1\left(\frac{l_{3,t}}{\|\boldsymbol{\rho_{1,t}}-\boldsymbol{\rho_{3,t}}\|}_2\!-\!1\right)^2}-3
		% 	\end{aligned}
	% \end{equation}
\begin{equation}
	\label{E19}
	R_{\text{sat-dis},t} = 
	\sum_{(i,j)\in \mathcal{N}}
	\frac{1}{1+\beta_3\left(\frac{l_{ij,t}}{\|\boldsymbol{\rho}_{i,t}-\boldsymbol{\rho}_{j,t}\|_2}-1\right)^2}
	-3,
\end{equation}
where $\beta_3>0$, $\mathcal{N}=\left\{\left(1,2\right),\left(2,3\right),\left(1,3\right)\right\}$.
\subsubsection{Episode Termination Conditions}
To guide the exploration behavior of RL and restrict its exploration scope, episode termination conditions are defined as( Eq. \eqref{E20} for the tethers' environment,Eq. \eqref{E21} for the node satellites' environment,):
\begin{equation}
	\label{E20}
	\|H_1\boldsymbol{e}_{l,t}\|_2>\vartheta_1\vee \boldsymbol{\eta}_t=\boldsymbol{\eta}_{final}
\end{equation}
\begin{equation}
	\label{E21}
	\|H_3\boldsymbol{e}_{sat,t}\|_2>\vartheta_2\vee \boldsymbol{\varepsilon}_{t}=\boldsymbol{\varepsilon}_{final}
\end{equation}
where $\boldsymbol{\eta}_{final}$ and $\boldsymbol{\varepsilon}_{final}$ areare the respective final states of $\boldsymbol{\eta}_{r}$ and $\boldsymbol{\varepsilon}_{r}$ 
\subsubsection{Environment Reset Conditions}
At the beginning of each episode, initial states $\boldsymbol{\eta}_0$, $\boldsymbol{\varepsilon}_0$ are defined as arbitrary points along the desired trajectories.To enhance robustness of agents, perturbations are introduced to initial states. Specifically, reset conditions are formulated as:
\begin{equation}
	\label{E22}
	\boldsymbol{\eta}_0=\boldsymbol{\eta}_{r,n_1}+\varrho_1
\end{equation}
\begin{equation}
	\label{E23}
	\boldsymbol{\varepsilon}_0=\boldsymbol{\varepsilon}_{r,n_2}+\varrho_2
\end{equation}
where $\boldsymbol{\eta}_{r,n_1}$, $\boldsymbol{\varepsilon}_{r,n_2}$ denote states on the reference trajectories at randomly selected time instants $n_1$, $n_2$, and $\varrho_1$, $\varrho_2$ represent small perturbation vectors.

\subsection{Theory and Implimentation of Soft Actor-Critic Algorithm}
The control policies of both the tether-level and the node-satellite-level are trained using the Soft Actor–Critic (SAC) algorithm. For the sake of brevity, in this section we use the generic notation to denote the observable state $\boldsymbol{s}_t$, action $\boldsymbol{a}_t$, and reward $R_t$ which correspond to $\left\{\boldsymbol{s}_{sat,t},\boldsymbol{s}_{l,t}\right\}$, $\left\{\boldsymbol{u}_{l,t}^{norm},\boldsymbol{\nu}_{l,t}^{norm}\right\}$, and $\left\{R_{sat,t},R_{l,t}\right\}$ respectively.
\subsubsection{Soft Policy Evaluation and Improvement}
The state-value function is defined  through a Bellman expectation equation:  
\begin{equation}
	\label{E24}
	V_{\pi}\left(\boldsymbol{s}_t\right)=\mathbb{E}_{\pi}\left[R_t+\gamma V\left(\boldsymbol{s}_{t+1}\right)\right]
\end{equation}
which represents the expected return attained when initiating in state $\boldsymbol{s}_t$ and following policy $\pi$.

Unlike conventional actor–critic (AC) algorithms, in SAC algorithm, an entropy regularization term is introduced to balance the exploration–exploitation trade-off in reinforcement learning. The action-value function is defined as:
\begin{equation}
	\label{E25}
	\left\{
	\begin{aligned}
		Q_{\pi}(\boldsymbol{s}_t,\boldsymbol{a}_t) 
		&=R_t\!+\!\gamma \, \mathbb{E}_{\boldsymbol{s}_{t+1}}\!\! \Big[ \!V_{\pi}(\boldsymbol{s}_{t+1}) 
		\!+\! \varpi \mathcal{H}\big(\pi(\cdot|\boldsymbol{s}_{t+1})\big)\!\! \Big] \\
		\mathcal{H}\big(\pi(\cdot|\boldsymbol{s}_{t+1})\big) 
		&\begin{aligned}[t]
			&=\!\! -\!\!\sum_{\boldsymbol{a}_{t+1}}\! \pi(\boldsymbol{a}_{t+1}|\boldsymbol{s}_{t+1}) 
			\log \pi(\boldsymbol{a}_{t+1}|\boldsymbol{s}_{t+1}\!) \\
			&=\!\! -\mathbb{E}_{\boldsymbol{a}_{t+1} \sim \pi(\cdot|\boldsymbol{s}_{t+1})} 
			\big[\!\log \pi(\boldsymbol{a}_{t+1}|\boldsymbol{s}_{t+1})\!\big]
		\end{aligned}
	\end{aligned}
	\right.
\end{equation}
where $\mathcal{H}\big(\pi(\cdot|\boldsymbol{s}_{t+1})\big)$ denotes The policy entropy at state $\boldsymbol{s}_{t+1}$, which quantifies the uncertainty of the policy distribution under state $\boldsymbol{s}_{t+1}$. $\varpi$ is a temperature parameter.

During the update phase of training, the policy evaluation and policy improvement are alternately executed. In policy evaluation, the Q-value is computed by applying a Bellman
operation $Q_{\pi}(\boldsymbol{s}_t,\boldsymbol{a}_t)=\mathcal{T}^{\pi}Q_{\pi}(\boldsymbol{s}_t,\boldsymbol{a}_t)$ \cite{2021Model} where:
\begin{equation}
	\label{E26}
	\begin{aligned}
		\mathcal{T}^{\pi}Q_{\pi}(\boldsymbol{s}_t,\boldsymbol{a}_t)=&R_t\!+\!\gamma\mathbb{E}_{\boldsymbol{s}_{t\!+\!1}}\!\!\left\{\mathbb{E}_{\boldsymbol{a}_{t\!+\!1}\sim\pi(\cdot|\boldsymbol{s}_{t\!+\!1})}\left[Q_{\pi}\!\!\left(\right.\right.\right.\\
		&\left.\left.\left.\!\!\!\!\boldsymbol{s}_{t\!+\!1},\boldsymbol{a}_{t\!+\!1}\right)\!-\!\varpi\log\!\left(\pi\left(\boldsymbol{a}_{t+1}|\boldsymbol{s}_{t+1}\!\right)\right)\right]\right\}
	\end{aligned}
\end{equation}
In policy improvement, the policy is updated through:
\begin{equation}
	\label{E27}
	\pi_{new}=\arg\min_{\pi'\in\Pi}\mathfrak{D}_{KL}\left(\pi'\left(\cdot|\boldsymbol{s}_t\right)\frac{e^{\varpi^{-1}Q^{\pi_{old}}(\boldsymbol{s}_t,\cdot)}}{Z^{\pi_{old}}}\right)
\end{equation}
where $\Pi$ denotes a policy set, $\mathfrak{D}$ repesents the KL divergence, $\pi_{old}$ denotes the policy from last policy improvement, $Q^{\pi_{old}}$ is the Q-value of $\pi_{old}$, and $Z^{\pi_{old}}$ is a normalization factor.

Through multiple iterations the policy eventually converges to the optimal one:
\begin{equation}
	\label{E28}
	\pi^\ast=\arg\min_{\pi\in\Pi}\mathbb{E}_{\pi}\left[\varpi\log\!\left(\pi\left(\boldsymbol{a}_{t}|\boldsymbol{s}_{t}\!\right)\right)-Q_{\pi}\left(\boldsymbol{s}_t,\boldsymbol{a}_t\right)\right]
\end{equation}
\subsubsection{Practical Implementation}
In SAC, both the action-value function and the policy are approximated by deep neural networks. To mitigate the issue of overestimation in Q-values, in addition to two Q-networks $Q_{\omega_1}$ and $Q_{\omega_2}$, two target-Q-networks $Q_{\overline{\omega}_1}$ and $Q_{\overline{\omega}_2}$ are employed and the minimum of their outputs is taken when computing the target. The critic networks are updated by minimizing the following loss function:
\begin{equation}
	\label{E29}
	\begin{aligned}
		L_Q(\omega_{p})\!=\!\mathbb{E}_{(\boldsymbol{s}_t,\boldsymbol{a}_t,R_t,\boldsymbol{s}_{t+1})\sim \mathcal{D}}\!\!&\left[\frac{1}{2}\!\!\left(Q_{\omega_p}(\boldsymbol{s}_t,\boldsymbol{a}_t)\right.\right.\\
		&\left.\left.-Y_{tar}\right)^2\right],p=1,2
	\end{aligned}
\end{equation}
where $\mathcal{D}$ denotes the replay buffer, $\omega_p$ represent parameters of Q-networks. The target value is:
\begin{equation}
	\label{E30}
	Y_{tar}\!=\!R_t\!+\!\gamma\min\left\{Q_{\overline{\omega}_1},Q_{\overline{\omega}_2}\right\}\!-\gamma\alpha\log\pi_{\theta}(\boldsymbol{a}_{t+1}|\boldsymbol{s}_{t+1})
\end{equation}
where $\overline{\omega}_1$, $\overline{\omega}_2$ are parameters of target-Q-networks, $\theta$ is the parameter of actor network.

The loss function of policy $\pi$ is obtained from the KL divergence:
\begin{equation}
	\label{E31}
	\begin{aligned}
		L_{\pi}(\theta)\!=\!\mathbb{E}_{\boldsymbol{s}_t\sim\mathcal{D},\boldsymbol{a}_t\sim\pi_{\theta}}&\Big[\varpi\log\pi_{\theta}(\boldsymbol{a}_{t}|\boldsymbol{s}_{t})\!\\
		&-\!\min_{p=1,2}Q_{\omega_p}(\boldsymbol{s}_{t},\boldsymbol{a}_{t})\Big]
	\end{aligned}
\end{equation}

The entropy term $\mathcal{H}$ serves as a measure of policy stochasticity: it should be larger in states with high uncertainty about the optimal action and smaller in states where the optimal action is nearly deterministic. Accordingly, the loss function for $\varpi$ is formulated as:
\begin{equation}
	\label{E32}
	L(\varpi)=\mathbb{E}_{\boldsymbol{s}_t\sim\mathcal{D},\boldsymbol{a}_t\sim\pi_{\theta}}\left[-\varpi\log\pi(\boldsymbol{a}_t|\boldsymbol{s}_t)-\varpi\mathcal{H}_0\right]
\end{equation}
where $\mathcal{H}_0$ is a target entropy.

The pseudo-code of SAC is shown in Algorithm \ref{alg:SAC}.
\begin{remark}\label{remark2}
	Policy $\pi_{\theta}(\boldsymbol{a}|\boldsymbol{s})$ is modeled as a Gaussian distribution $\mathcal{N}(c_{\theta}(\boldsymbol{s}),\mu_{\theta}(\boldsymbol{s}))$, where $c_{\theta}(\boldsymbol{s})$ and $\mu_{\theta}(\boldsymbol{s})$ are outputs of the critic network. During training, an action $a_t$ is sampled from the policy distribution, and entropy regularization is applied to encourage sufficient exploration. While testing, the agent take the mean value $a_t=c_{\theta}(\boldsymbol{s})$ for deterministic control. Therefore, the controller can guarantee system stability
\end{remark}
\begin{algorithm}[h]
	\caption{Soft Actor-Critic (SAC)}
	\label{alg:SAC}
	\KwIn{Learning rates $\varkappa$, discount factor $\gamma$, temperature $\varpi$, total training time step $K$, replay buffer $\mathcal{D}$}
	\KwOut{Trained policy $\pi_\theta$}
	
	Randomly initialize critic networks $Q_{\omega_1}(\boldsymbol{s},a)$, $Q_{\omega_2}(\boldsymbol{s},a)$ and actor network $\pi_\theta(\boldsymbol{s})$\;
	Initialize target networks $\bar{Q}_{\omega_1} \leftarrow Q_{\omega_1}, \; \bar{Q}_{\omega_2} \leftarrow Q_{\omega_2}$\;
	Initialize replay buffer $\mathcal{D}$\;
	
	\For{each episode}{
		Obtain initial state $\boldsymbol{s}_1$\;
		\For{each time step $t$}{
			Sample action $\boldsymbol{a}_t \sim \pi_\theta(\cdot|\boldsymbol{s}_t)$\;
			Execute $a_t$, receive reward $r_t$, and observe next state $\boldsymbol{s}_{t+1}$\;
			Store transition $(\boldsymbol{s}_t, \boldsymbol{a}_t, R_t, \boldsymbol{s}_{t+1})$ in $\mathcal{D}$\;
			
			\For{each gradient step}{
				Sample minibatch $\{(\boldsymbol{s}_q,\boldsymbol{a}_q,R_q,\boldsymbol{s}_{q+1})\}_{q=1}^N$ from $\mathcal{D}$\;
				Compute target value for each sample:\
				\begin{align*}
					y_q &= R_q + \gamma \bigg(
					\min_{p=1,2} Q_{\bar{\omega}_p}(\boldsymbol{s}_{q+1}, \boldsymbol{a}_{q+1}) \\
					&\quad - \varpi \log \pi_\theta(a_{q+1}|s_{q+1})
					\bigg), \\
					&a_{q+1} \sim \pi_\theta(\cdot|s_{q+1})
				\end{align*}
				Update critic networks by minimizing:\
				\[
				L(\omega_p) \!=\! \frac{1}{N}\sum_{q=1}^N \big( y_q - Q_{\omega_p}(\boldsymbol{s}_q,\boldsymbol{a}_i) \big)^2, p\!=\!1,2
				\]
				Update the actor network by minimizing:\
				\begin{align*}
					L_\pi(\theta) &= \frac{1}{N}\sum_{q=1}^N 
					\Big( \alpha \log \pi_\theta(\tilde{\boldsymbol{a}}_q|s_q)\\
					&- \min_{p=1,2} Q_{\omega_p}(\boldsymbol{s}_q, \tilde{\boldsymbol{a}}_q) \Big),
					\tilde{\boldsymbol{a}}_q \sim \pi_\theta(\cdot|\boldsymbol{s}_i)
				\end{align*}
				Update temperature parameter $\varpi$\;
				Update target networks:\
				\begin{align*}
					\bar{\omega}_p \leftarrow \varkappa \omega_p + (1-\varkappa)\bar{\omega}_p,p=1,2
				\end{align*}
				\textbf{end for}
			}
			\textbf{end for}
		}
		\textbf{end for}
	}
\end{algorithm}

\section{Stability of the Tracking Control}
\label{Sec:sec5}
\subsection{Tether-Level Analysis}
\subsubsection{Baseline Control Law: Proof of Uniform Ultimate Boundedness (UUB)}
The error dynamics equation of \eqref{E6} can be formulated as:
\begin{equation}
	\label{E33}
	\dot{\boldsymbol{e}}_l=\Upsilon_m\boldsymbol{\eta}+\Xi_m\boldsymbol{\nu}+\Delta_1-\dot{\boldsymbol{\eta}}_r
\end{equation}
where $\Upsilon_m=
\left[
\begin{array}{cc}
	0 & I_3 \\
	0 & -\frac{1}{T_m}I_3 
\end{array}
\right]$, $\Xi_m=\left[
\begin{array}{c}
	0 \\
	\frac{D_r}{T_mK_e}I_3
\end{array}
\right]$, $\Delta_1=\left[
\begin{array}{c}
	0 \\
	\frac{D_r}{T_m}I_3
\end{array}
\right]\boldsymbol{\delta}$.
\begin{theorem}
	\label{theo1}
	Under the following condition, the baseline control law $\boldsymbol{\nu}_b$ guarantees that the system state is globally uniformly bounded.
	\begin{enumerate}
		\item The control law is given by
		\begin{equation}\label{E34}
			\boldsymbol{\nu}_b = -K_1\boldsymbol{e}_l
		\end{equation}
		\item There exists a symmetric positive definite matrix $P_1$ such that:
		\begin{equation}\label{E35}
			\begin{aligned}
				&(\!\Upsilon_m\!\!-\!\Xi_m K_1\!)^T\! P_1 \!+\! P_1\!(\Upsilon_m\!-\!\Xi_m K_1) \!=\! -\tilde{Q}_1,\\
				&{\tilde{Q}_1} \succ 0 , \quad {\tilde{Q}_1}^T={\tilde{Q}_1}
			\end{aligned}
		\end{equation}
	\end{enumerate}
\end{theorem}
\begin{proof}
	Substituting \eqref{E34} into \eqref{E33} the error dynamics can be written as:
	\begin{equation}
		\label{E36}
		\begin{aligned}
			&\dot{\boldsymbol{e}}_l = (\Upsilon_m-\Xi_m K_1){\boldsymbol{e}}_l + \tilde{\Delta}_1\\
			&\tilde{\Delta}_1=\Upsilon_m\boldsymbol{\eta}_r-\dot{\boldsymbol{\eta}}_r+\Delta_1
		\end{aligned}
	\end{equation}
	Consider the Lyapunov candidate function
	\begin{equation}
		\label{E37}
		V_1(\boldsymbol{e}_l) = \boldsymbol{e}_l^T P_1 \boldsymbol{e}_l , \quad P_1 \succ 0 .
	\end{equation}
	The time derivative of $V(\boldsymbol{e}_l)$ along the trajectories is
	\begin{equation}
		\label{E38}
		\begin{aligned}
			\dot{V}_1(\boldsymbol{e}_l) =& \boldsymbol{e}_l^T \big( (\Upsilon_m\!-\!\Xi_m K_1)^T P_1 \\
			&+P_1(\Upsilon_m\!-\!\Xi_m K_1) \big) \boldsymbol{e}_l + 2 \boldsymbol{e}_l^T P_1 \tilde{\Delta}_1 \\
			=& - \boldsymbol{e}_l^T \tilde{Q}_1 \boldsymbol{e}_l + 2 \boldsymbol{e}_l^T P_1 \tilde{\Delta}_1
		\end{aligned}
	\end{equation}
	Since $P_1$ and $\tilde{Q}_1$ are symmetric positive definite matrixs, we have:
	\begin{equation}
		\label{E39}
		\begin{aligned}
			&\lambda_{min}(\tilde{Q}_1)\| \boldsymbol{e}_l \|_2^2\leq\boldsymbol{e}_l^T\tilde{Q}_1\boldsymbol{e}_l\leq \lambda_{max}(\tilde{Q}_1)\| \boldsymbol{e}_l \|_2^2\\
			& \|P_1\|_2=\lambda_{max}(P_1)
		\end{aligned}
	\end{equation}
	Applying the Cauchy--Schwarz inequality yields:
	\begin{equation}
		\label{E40}
		\dot{V}(\boldsymbol{e}_l) \leq -\lambda_{\min}(\tilde{Q}_1) \|\boldsymbol{e}_l\|_2^2 + 2\lambda_{max}(P_1)\bar{\Delta}_1\|\boldsymbol{e}_l\| _2
	\end{equation}
	
	Introduce $c_1 = \frac{2\lambda_{max}(P_1)\bar{\Delta}_1}{\lambda_{min}(\tilde{Q}_1)}$. When $\|\boldsymbol{e}_l\|_2\geq c_1$, $\dot{V}_1(\boldsymbol{e}_l)<0$ . 
	
	Therefore, the error system is uniformly ultimately bounded with ultimate bound to $c_1$.
\end{proof}
\subsubsection{RL-based Control Law: Enhancement over the Baseline Control}
Since the baseline control $\boldsymbol{\nu}_b$ is admissible, we start the RL training process using this baseline controller.
\begin{theorem}\label{theo2}
	The overall control law $\boldsymbol{\nu}^h = \boldsymbol{\nu}_b +\boldsymbol{\nu}_l^h$
	can always stabilize the reel motor system \eqref{E6}, where $\boldsymbol{\nu}_l^h$ represents RL-based Control Law from the h-th iteration, and $h=1,2,3,\ldots\infty$
\end{theorem}
\begin{proof}
	The initial RL control and  the associated action-value function are denoted by $\boldsymbol{\nu}^0_l$ and $Q^0(\boldsymbol{s}_{l},\boldsymbol{\nu}^0_l)$. 
	Since the system dynamics are deterministic and no exploration noise is considered, 
	the soft $Q$-function reduces to the standard discounted return:
	\begin{equation}
		\label{E41}
		Q^0(\boldsymbol{s}_t,\boldsymbol{\nu}^0_{l,t}) = R^0_{l,t} + \gamma Q^0(\boldsymbol{s}_{t+1},\boldsymbol{\nu}^0_{l,t+1})
	\end{equation}
	
	According to the reward functions design in Section \uppercase\expandafter{\romannumeral4}-C-\ref{Sec:reward},  the reward is non-positive and bounded, i.e. $R_{l,t} \leq 0$ and $R_{l,t} \in [R_{l}^{min},0]$. 
	Thus $Q^0 \leq 0$ and bounded. 
	
	Define the Lyapunov candidate:
	\begin{equation}\label{E42}
		V_1^0(\boldsymbol{e}_l) = -Q^0(\boldsymbol{s}_{l},\boldsymbol{\nu}^0_l) \geq 0,
	\end{equation}
	Using the Bellman relation, we obtain:
	\begin{equation}\label{E43}
		V_1^0(\boldsymbol{e}_{l,t}) = -R^0_{l,t} + \gamma V_1^0(\boldsymbol{e}_{l,t+1}).
	\end{equation}
	
	From LEMMA~1, the baseline control $\boldsymbol{\nu}_b$ is admissible, which guarantees that there exists a continuously differentiable function $V(\cdot)$ such that:
	\begin{equation}
		\label{E44}
		V(\boldsymbol{e}_{l,t+1}) - V(\boldsymbol{e}_{l,t}) \leq -W(\boldsymbol{e}_{l,t}) + \psi_3(\|\tilde{\Delta}_1(t)\|_2^2)
	\end{equation}
	with $W(\boldsymbol{e}_{l,t})$ positive definite and dominating the uncertainty term when $\|\boldsymbol{e}_{l,t}\|$ is large. 
	Comparing with \eqref{E43}, one sees that:
	\begin{equation}
		\label{E45}
		(1\!-\!\gamma)V_1^0(\boldsymbol{e}_{l,t+1})\!+\!R_{l,t}^0\leq -W(\boldsymbol{e}_{l,t}) \!+\! \psi_3(\|\tilde{\Delta}_1(t)\|_2^2)
	\end{equation}
	This inequality implies that the change in the Lyapunov function is bounded by the system uncertainty. Thus, the system under the initial control law $\boldsymbol{\nu}^0 = \boldsymbol{\nu}_b + \boldsymbol{\nu}_l^0$ is UUB.
	
	Next, consider one step of policy improvement. 
	The new control $\boldsymbol{\nu}_l^1$ is obtained by minimizing:
	\begin{equation}\label{E46}
		\boldsymbol{\nu}_l^1 = \arg\min_{\boldsymbol{\nu}_l}\bigl\{-R^0_{l,t} + \gamma V_1^0(\boldsymbol{e}_{l,t+1 })\bigr\}
	\end{equation}
	In policy evaluation the value function updates as:
	\begin{equation}\label{E47}
		V_1^1(\boldsymbol{e}_{l,t+1}) = -R^1_{l,t} + \gamma V_1^0(\boldsymbol{e}_{l,t+1})
	\end{equation}
	and finally converge at:
	\begin{equation}\label{E48}
		V_1^1(\boldsymbol{e}_{l,t}) = -R^1_{l,t} + \gamma V_1^1(\boldsymbol{e}_{l,t+1})
	\end{equation}
	Since $V_1^1(\boldsymbol{e}_{l,t})\leq V_1^0(\boldsymbol{e}_{l,t})$, for $\boldsymbol{\nu}_l^1$, there exsists:
	\begin{equation}\label{E49}
		\begin{aligned}
			V_1^{1}(\boldsymbol{e}_{l,t+1}) - V^{1}_1(\boldsymbol{e}_{l,t}) \!=&V_1^{1}(\boldsymbol{e}_{l,t+1})+\!R^1_{l,t}\\ 
			&-\gamma V_1^{1}(\boldsymbol{e}_{l,t+1})\\
			\leq&-\!W(\boldsymbol{e}_{l,t})\! +\! \psi_3(\|\bar{\Delta}_1(t)\|_2)\\
			&- (R_{l,t}^0 - R_{l,t}^{1})
		\end{aligned}
	\end{equation}
	Since $R_{l,t}^{1} \geq R_{l,t}^0$ (due to policy improvement maximizing cumulative reward), the term $-(R_{l,t}^0 - R_{l,t}^{1})$ is non-positive. Furthermore, both $R_{l,t}^0$ and $R_{l,t}^{1}$ are bounded. Therefore, there exists a new constant $\bar\Delta_1^{1} > 0$ such that:
	\begin{equation}\label{E50}
		W(\boldsymbol{e}_{l,t}) > \psi_3(\|\Delta(t)\|_2) - (R_t^0 - R_t^{1}),  \forall \|\boldsymbol{e}_{l,t}\|_2 > \bar\Delta_1^{1}
	\end{equation}
	This ensures that $V_1^{1}(\boldsymbol{e}_{l,t+1}) - V^{1}_1(\boldsymbol{e}_{l,t}) < 0$ outside a bounded set, confirming that the system under the improved policy $\boldsymbol{\nu}^1 = \boldsymbol{\nu}_b +\boldsymbol{\nu}_l^1$ is also UUB.
	
	Repeating the same analysis for each iteration $h$ , it is poven inductively that every $\boldsymbol{\nu}^h$ stabilizes the system. 
	Thus, the overall control law $\boldsymbol{\nu}^h = \boldsymbol{\nu}_b +\boldsymbol{\nu}_l^h$ guarantees UUB stability for all iterations. 
\end{proof}
\subsection{Satellite-Level Analysis}
\subsubsection{Boundedness Analysis of Tension}
Denote the reference position and the position error of the $i$th satellite as: $\boldsymbol{r}_i^{ref}$ and $\Delta{\boldsymbol{r}}_i$ respectively.
\begin{equation}
	\label{E51}
	\begin{aligned}
		\|\boldsymbol{r}_i-\boldsymbol{r}_j\|_2=&\|\boldsymbol{r}_{i}^{ref}-\boldsymbol{r}^{ref}_j+\Delta{\boldsymbol{r}_{i}}-\Delta{\boldsymbol{r}_{j}}\|\\
		\leq&\|\boldsymbol{r}_{i}^{ref}-\boldsymbol{r}^{ref}_j\|_2+\|\Delta{\boldsymbol{r}_{i}}-\Delta{\boldsymbol{r}_{j}}\|_2
	\end{aligned}
\end{equation}
Denote length error as $\Delta_{ij}$. After training in tether-level, THEOREM \ref{theo2} holds, thus:
\begin{equation}
	\label{E52}
	\|\boldsymbol{r}_{i}^{ref}-\boldsymbol{r}^{ref}_j\|_2=l_{ij}+\Delta{ij}
\end{equation}
Substituting Eq. \ref{E52} in Eq. \ref{E51}:
\begin{equation}\label{E53}
	\|\boldsymbol{r}_i-\boldsymbol{r}_j\|_2-l_{ij}\leq \|\Delta{\boldsymbol{r}_{i}}-\Delta{\boldsymbol{r}_{j}}\|_2+\Delta_{ij}
\end{equation}
Therefore, tension terms satisfy:
\begin{equation}\label{E54}
	\begin{aligned}
		\|\boldsymbol{T}_{ij}\|_2&\leq \frac{K_T}{l_{ij}}(\|\Delta{\boldsymbol{r}_{i}}-\Delta{\boldsymbol{r}_{j}}\|_2+\Delta_{ij})\\
		&\leq c_2\|\boldsymbol{e}_{sat}\|_2+c_3
	\end{aligned}
\end{equation}
where $c_2=2K_T/l^{min}$, $c_3=K_T\|\Delta{ij}\|_2$.
\subsubsection{Baseline Control Law: Proof of Uniform Ultimate Boundedness (UUB)}
The error dynamics \eqref{E4} can be formulated as:
\begin{equation}
	\label{E55}
	\dot{\boldsymbol{e}}_{sat}=\tilde{\Upsilon}_m\boldsymbol{\varepsilon}+\tilde{\Xi}_m\boldsymbol{u}+\tilde{\Delta}_2-\dot{\boldsymbol{\varepsilon}}_r
\end{equation}
where $\tilde{\Upsilon}_m=
\left[
\begin{array}{cc}
	0 & I_6 \\
	A_m & I_3 \otimes C
\end{array}
\right]$, $\tilde{\Xi}_m=\left[
\begin{array}{c}
	0 \\
	I_6
\end{array}
\right]$, $\tilde{\Delta}_2=\left[
\begin{array}{c}
	0 \\
	\boldsymbol{G}(\boldsymbol{r}_{sat}) \!+\! \boldsymbol{D}
\end{array}
\right]$.
\begin{theorem}
	\label{theo3}
	Under the following condition, the baseline control law $\boldsymbol{u}_b$ guarantees that the system state is globally uniformly bounded (UUB).
	\begin{enumerate}
		\item The control law is given by
		\begin{equation}\label{E56}
			\boldsymbol{u}_b = -K_2\boldsymbol{e}_{sat}
		\end{equation}
		\item There exists a symmetric positive definite matrix $P_2$ such that:
		\begin{equation}\label{E57}
			\begin{aligned}
				&(\!\tilde{\Upsilon}_m\!-\!\Xi_m K_2\!)^T\! P_2 \!+\! P_2(\tilde{\Upsilon}_m\!-\!\Xi_m K_2\!)\! =\! -\tilde{Q}_2,\\
				&2c_2\lambda_{max}{P_2}-\lambda_{\min}(\tilde{Q}_2)>0\\
				&{\tilde{Q}_2} \succ 0 , \quad {\tilde{Q}_2}^T={\tilde{Q}_2}
			\end{aligned}
		\end{equation}
	\end{enumerate}
\end{theorem}
\begin{proof}
	Considier the Lyapnov candidate function:
	\begin{equation}\label{E58}
		V_2(\boldsymbol{e}_{sat})=\boldsymbol{e}_{sat}^TP_2\boldsymbol{e}_{sat}
	\end{equation}
	The derivative of $V_2$ can be obtained using the same technique as the differentiation of $V_1$:
	\begin{equation}\label{E59}
		\begin{aligned}
			\dot{V}_2\leq& -\lambda_{\min}(\tilde{Q}_2)\|\boldsymbol{e}_{sat}\|_2^2+2\lambda_{\max}(P_2)\bar{\Delta}_2\|\boldsymbol{e}_{sat}\|_2\\
			&+2c_2\lambda_{max}(P_2)\|\boldsymbol{e}_{sat}\|_2^2
		\end{aligned}
	\end{equation}
	Introduce $c_3 = \frac{2\lambda_{max}(P_2)\bar{\Delta}_2}{\lambda_{min}(\tilde{Q}_2)-2c_2\lambda_{max}(P_2)}$. When $\lambda_{min}(\tilde{Q}_2)-2c_2\lambda_{max}(P_2)>0,\|\boldsymbol{e}_{sat}\|_2\geq c_3$, $\dot{V}_2(\boldsymbol{e}_{sat})<0$ .
	
	Therefore, the error system is uniformly ultimately bounded with ultimate bound to $c_3$.
\end{proof}
\subsubsection{RL-based Control Law: Enhancement over the Baseline Control}
\begin{theorem}\label{theo4}
	The overall control law $\boldsymbol{u}^{\tilde{h}} = \boldsymbol{u}_b +\boldsymbol{u}_l^{\tilde{h}}$
	can always stabilize system \eqref{E4}, where $\boldsymbol{u}_l^{\tilde{h}}$ represents RL-based Control Law from the $\tilde{h}$-th iteration, and $\tilde{h}=1,2,3,\ldots\infty$
\end{theorem}
\begin{proof}
	The proof follows the same lines of THEOREM \ref{theo2}
\end{proof}

\section{SIMULATION VALIDATION}
\label{Sec:sec6}
To validate the the effectiveness and performance of the control scheme we proposed, numerical simulations are carried out. 

The desired trajectory of $l_i$ is designed as:
\begin{equation}
	\label{E60}
	l_d=\int_{t_0}^{t_f}\chi_d{\rm d}t+l_0
\end{equation}
with $\chi_d$, the desired release rate of $l_i$ designed by:
\begin{equation}
	\label{E61}
	\chi_d = a
\end{equation}
$a$ is an optional design constant and $l_0$ is the initial condition of $l_d$.

%According to \cite{zhang2021stable,zhou2020Dynamics}, TTFS can reach a spin-stablized state if the rotation angular velocity $\omega_0$ and the orbital angular velocity $n$ of TTFS satisfy the condition: $|{\omega_0}/{n}|>\vert\sqrt{{5}/{2}}-1\vert$. Specifically, in that case, satellites rotate around the system centriod while the system centriod moves on the circular oribit. Fuel consumption on configuration keeping can be significantly reduced  because taut tethers play the roll of geometric constraints. Thus, the desired position of the $i$th satellite is described as:
Under the reference trajectory, TTFS always retains an equilateral triangle structure while rotating and expanding. The reverse solution, using geometric relations, leads to satellites' reference track as:
\begin{equation}
	\label{E62}
	\left\{\begin{aligned}
		x_{i}^d=&\frac{l_d}{\sqrt{3}} \cos\left[w_0t+\frac{2\pi}{3}\left(i-1\right)\right]\\
		y_{i}^d=&\frac{l_d}{\sqrt{3}} \sin\left[w_0t+\frac{2\pi}{3}\left(i-1\right)\right]\\
	\end{aligned}\right.
\end{equation}
where $i=1,2,3$.
%Meanwhile, to keep tethers in micro-tensioned states, the desired length of tethers is designed as:
%\begin{equation}
%	\label{E55}
%	l_d=\left\{
%	\begin{aligned}
	%		&\sqrt{3}r_d-0.01\quad &&0\leq t<t_1 \\
	%		&\sqrt{3}r_d-0.0045 &&t>t_2\\
	%		&\sqrt{3}r_d-0.001 &&others			
	%	\end{aligned}
%	\right.
%\end{equation}
%\begin{equation}
%	\label{E55}
%l_d=\left(\sqrt{3}-0.0054\right)r_d
%\end{equation}

%Tethers are assumed to made of lightweight nylon for its strong tenacity and corrosion resistance. 
System parameters are shown in Table \ref{Table1:system parameters}. Configurations for the training and neural networks are found in Table \ref{Table2:learning parameters}
%Set the initial conditions of state vectors to: $\boldsymbol{r}_0^o=\left[\cos0,\sin0,\right.\\\left.(r_0-0.09)\cos\left({2\pi}/{3}\right), (r_0\!-\!0.09)\sin\left({2\pi}/{3}\right),(r_0\!-\!0.1)\right.\\\left.\cos\left({4\pi}/{3}\right),(r_0\!-\!0.1)\sin\left({4\pi}/{3}\right)\right]^T$, $\boldsymbol{v}_0^o=\omega_0\left[-\sin0,\right.\\\left.\cos0,-\sin\left({2\pi}/{3}\right),\cos\left({2\pi}/{3}\right),-\sin\left({4\pi}/{3}\right),\left({4\pi}/{3}\right)\right]^T$. 
The external disturbance is specified as: $\boldsymbol{\delta}=\left[\dot{l}_1\sin l_1,\dot{l}_2\sin(l_2),\dot{l}_3\sin(l_3)\right]^T$, $\boldsymbol{d}_i=0.1[\sin t,\cos t]^T$. 

Fig. \ref{Fig9}-\ref{Fig10} illustrate training curves of the proposed algorithm. Episode timesteps of tethers' environment start to increase after approximately 3600 episodes and begin to converge at around 3800 episodes, with the episode reward stabilizing at about -200.
In the satellites environment, the episode timesteps begin to rise after roughly 1,700 episodes and converge at around 1,800 episodes, where the episode reward stabilizes at approximately -3800. After 1800 episodes, there were some fluctuations in the training performance. 
As mentioned in remark \ref{remark2}, actions are sampled from policy distributions while training the resulting trajectories can exhibit fluctuations in episode length or reward, even when the underlying policy has already converged. In testing stage, mean values are taken for deterministic control, which naturally produces stable and consistent performance. This discrepancy between training and testing performance is a typical behavior of SAC \cite{haarnoja2018soft,zheng2022soft,tang2023path} and does not necessarily indicate poor policy learning.

Meanwhile, training curves under different reset conditions are plotted in Fig. \ref{Fig9}-\ref{Fig10} with blue dash lines. Compared with setting the environment reset condition as fixed points of desired trajectories, the random reset condition proposed in this paper accelerates the convergence speed of training in tethers' environment by approximately a factor of two. This improvement can be attributed to broader state-space coverage provided by random reset conditions, which enhance the agent’s exploration. In satellites' environment, when using the fixed-point reset condition, the optimal control policy still fails to emerge even by the end of training.

Fig. \ref{Fig11} gives training curves under centralized training and unnormalized-state training. Unlike the hierarchical training framework, centralized training framework suffers from severe coupling among system states which makes it difficult for the  agent to efficiently explore effective actions and thus leads to poor convergence performance. As shown by the blue dashed lines in Fig. \ref{Fig11}, training with unnormalized states fails to converge. Because the range of variation in tether length and release speed differs by two orders of magnitude, those two features exert highly unbalanced effects on networks. Since neural networks are sensitive to feature scales, such discrepancies hinder convergence during training.
\begin{table}[!htbp]
	\centering
	\caption{system parameters}\label{Table1:system parameters}
	\begin{tabular}{c c | c c}
		\toprule[1pt]
		Parameters      & Values  &Parameters &Values\\ \hline
		\makecell[c]{$m_{i}$/kg}   & \makecell[c]{50}   &$E$/($\rm Gpa$)	&$1.528$\\
		\makecell[c]{$R$/km} & \makecell[c]{7378}    &$A$/($\rm m^2$) &$1.963\times 10^{-7}$\\
		\makecell[c]{$\mu$/($\rm m^3/s^2$)}   & \makecell[c]{$3.98603\times10^{14}$} &$D_r$/m	&0.05\\
		\makecell[c]{$\omega_0$/(rad/s)}  &\makecell[c]{$0.003$}	&$R_a$/Ohms	&0.062\\
		\makecell[c]{$l_0$/m}   &\makecell[c]{1}	&$k_e$/(V/rad/s) &0.275\\
		$k_m$/(N$\cdot$ m/Amp)	&0.275	&\makecell[c]{ $J$/(kg$\cdot \rm m^2$)}
		&\makecell[c]{0.1}\\
		$a$/(m/s)	&1	&$l_0$/m	&1\\		
		\bottomrule[1.5pt]
	\end{tabular}
\end{table}

\begin{table}[!htbp]
	\centering
	\caption{learning parameters}\label{Table2:learning parameters}
	\begin{tabular}{c c}
		\toprule[1pt]
		Parameters      & Values  \\ \hline
		\makecell[c]{sample batch size $N$} & \makecell[c]{256}   \\
		Replay buffer size	&$10^6$\\
		\makecell[c]{total training timestep $K$}   & \makecell[c]{$10^7$}   \\
		\makecell[c]{learning rate$\varkappa$}   & \makecell[c]{$3\times10^{-5}$} \\
		\makecell[c]{$[\alpha_1,\alpha_2,\alpha_3,\alpha_4,\beta_1,\beta_2,\beta_3]$}  &\makecell[c]{$[1,2.5,1,7,1,1.5,100]$}	\\
		\makecell[c]{$\gamma$}   &\makecell[c]{0.99}	\\		
		\bottomrule[1pt]
	\end{tabular}
\end{table}

\begin{figure}[!htbp] 
	\centerline{\includegraphics[width=0.48\textwidth]{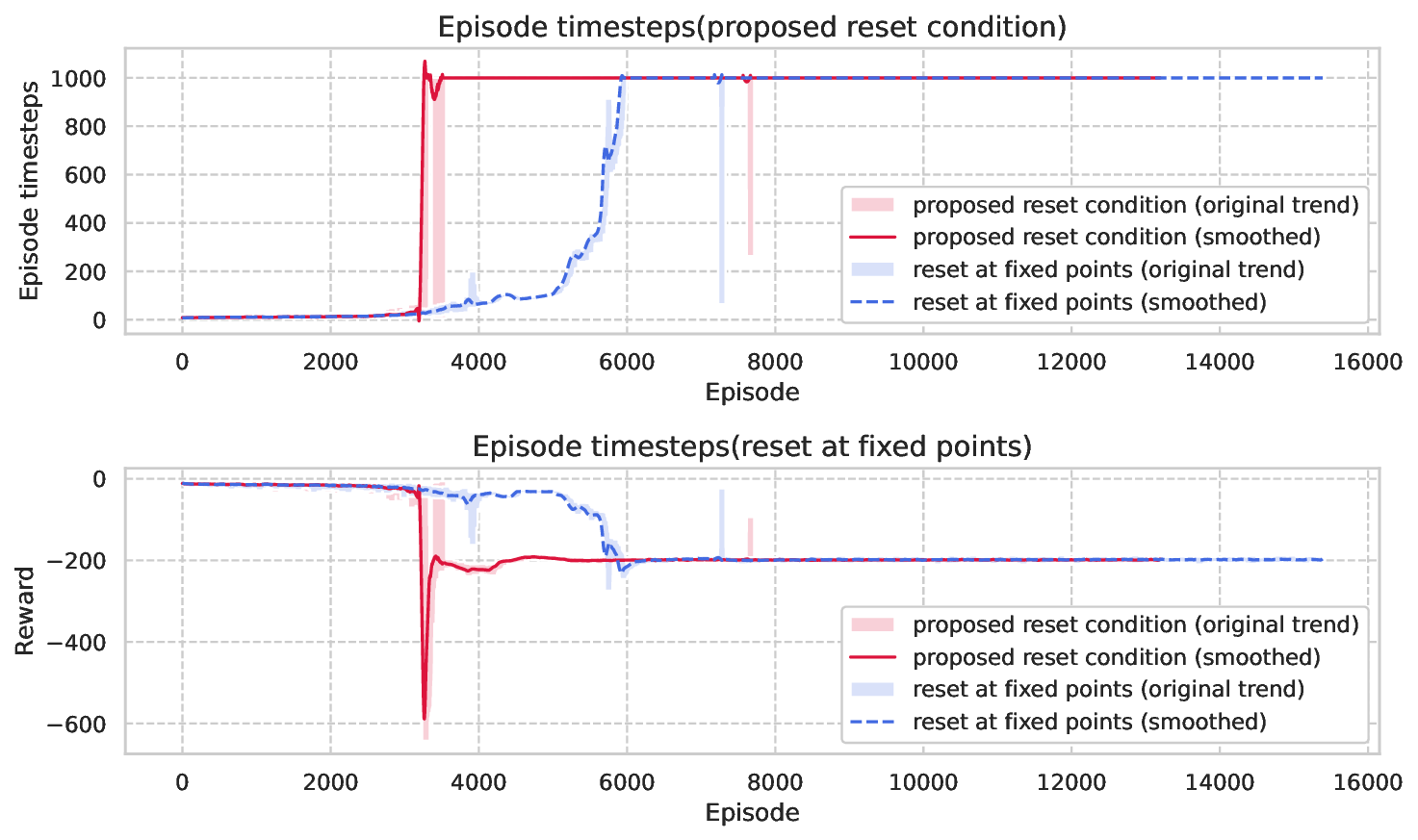}}
	\caption{Training curves of tethers}
	\label{Fig9}
\end{figure}

\begin{figure}[!htbp] 
	\centerline{\includegraphics[width=0.48\textwidth]{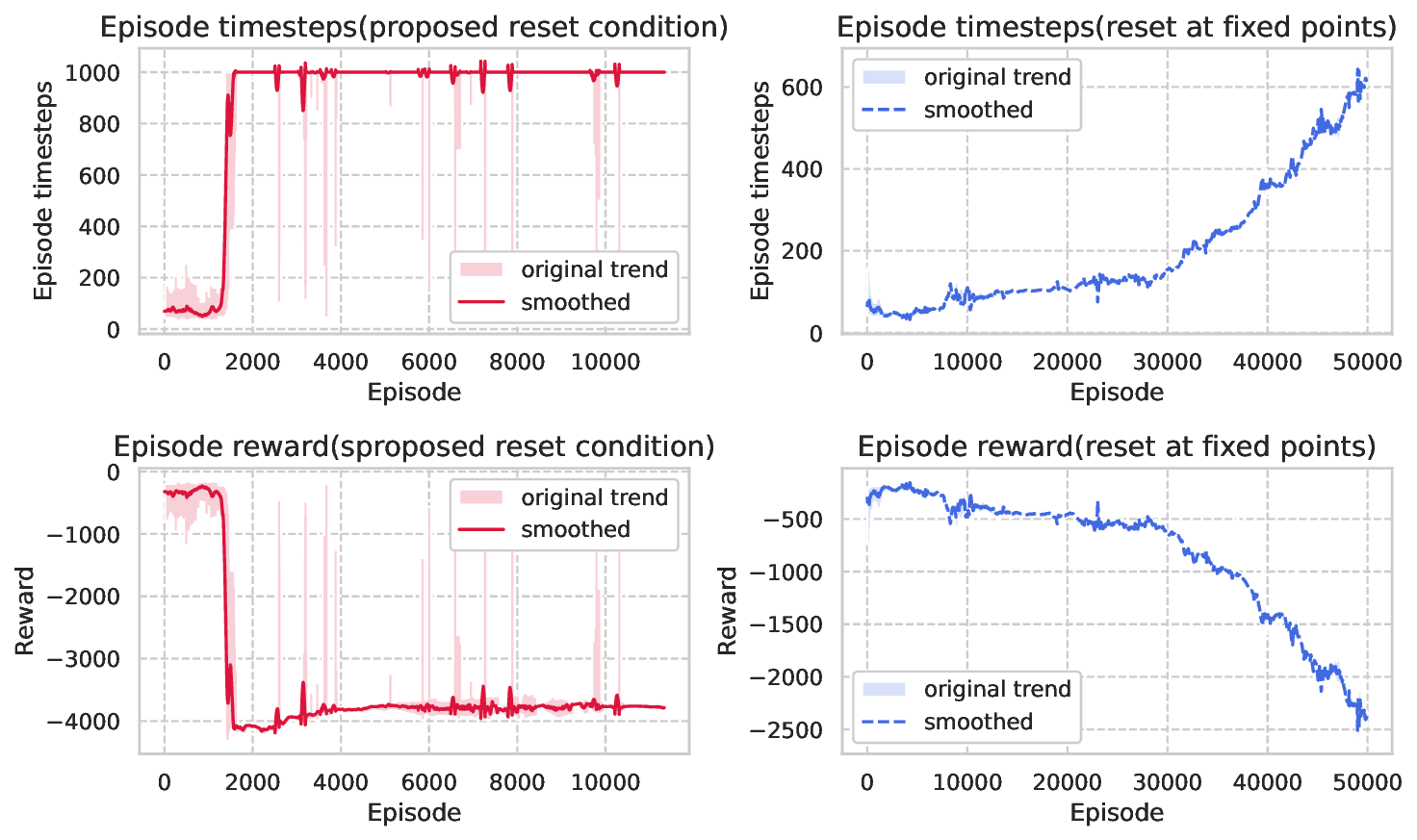}}
	\caption{Training curves of satellites}
	\label{Fig10}
\end{figure}

\begin{figure}[!htbp] 
	\centerline{\includegraphics[width=0.48\textwidth]{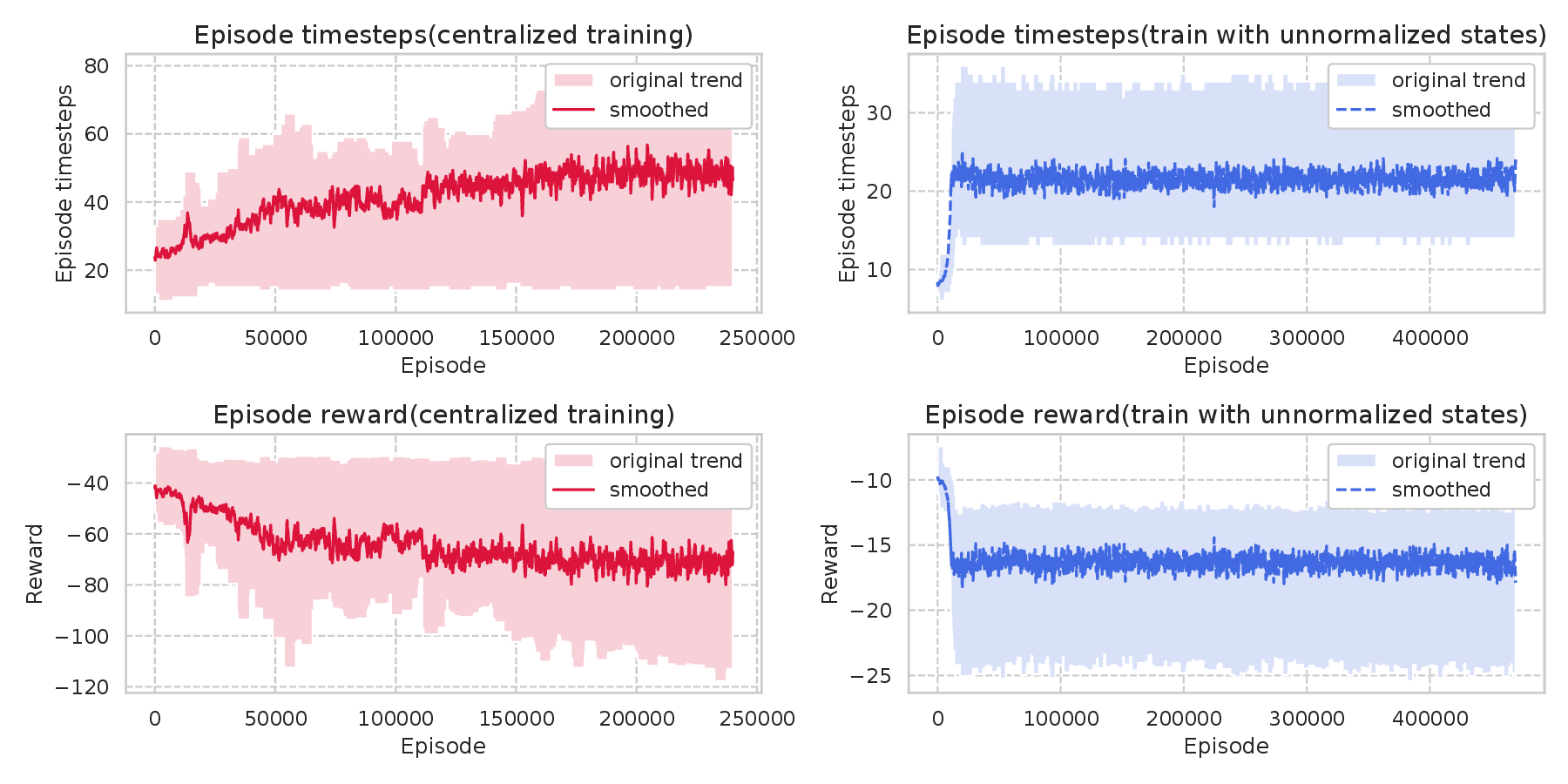}}
	\caption{Training curves under centralized training and unnormalized-state training}
	\label{Fig11}
\end{figure}

%\begin{figure}[htbp]
%	\centering
%	\includegraphics[width=0.48\textwidth, clip, trim=0mm 0mm 0mm 0mm]{pictures/training_curve_comp.eps}
%	\caption{Training curves comparison}
%	\label{fig:training_curves}
%\end{figure}

\begin{figure}[!htbp] 
	\centerline{\includegraphics[width=0.48\textwidth]{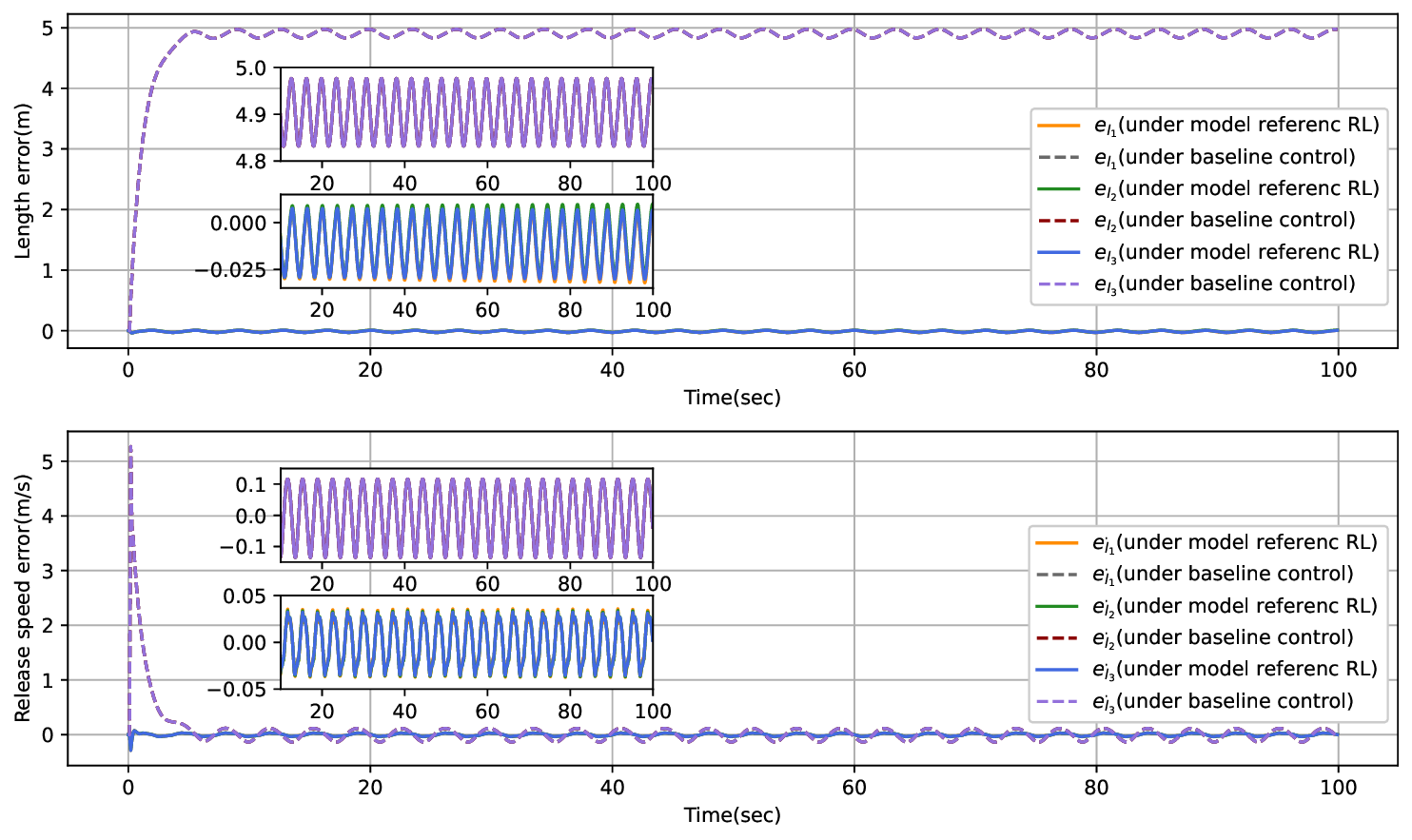}}
	\caption{Tracking error of tethers}
	\label{Fig4}
\end{figure}

\begin{figure}[!htbp] 
	\centerline{\includegraphics[width=0.48\textwidth]{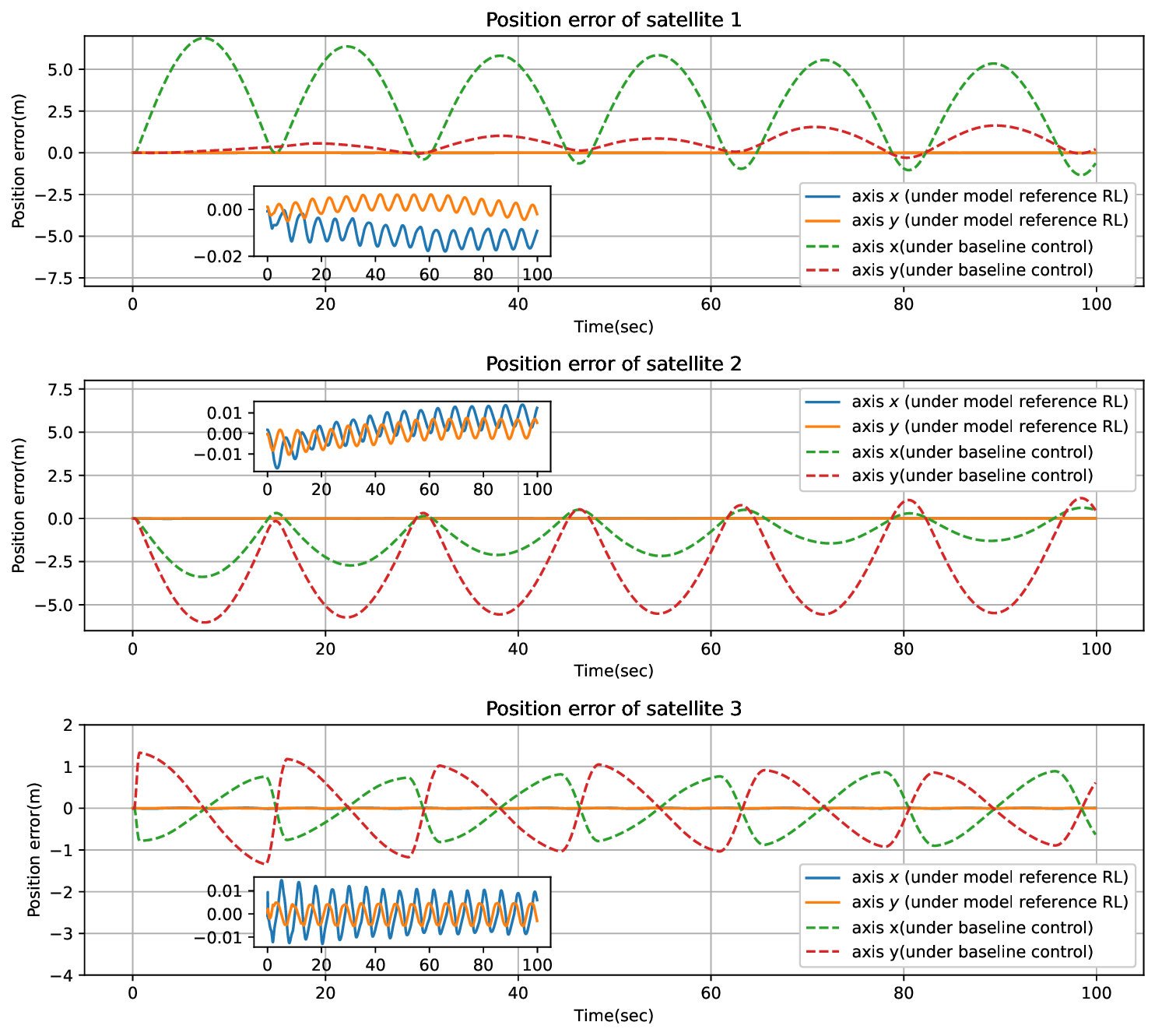}}
	\caption{Position error of satellites}
	\label{Fig5}
\end{figure}

\begin{figure}[!htbp] 
	\centerline{\includegraphics[width=0.48\textwidth]{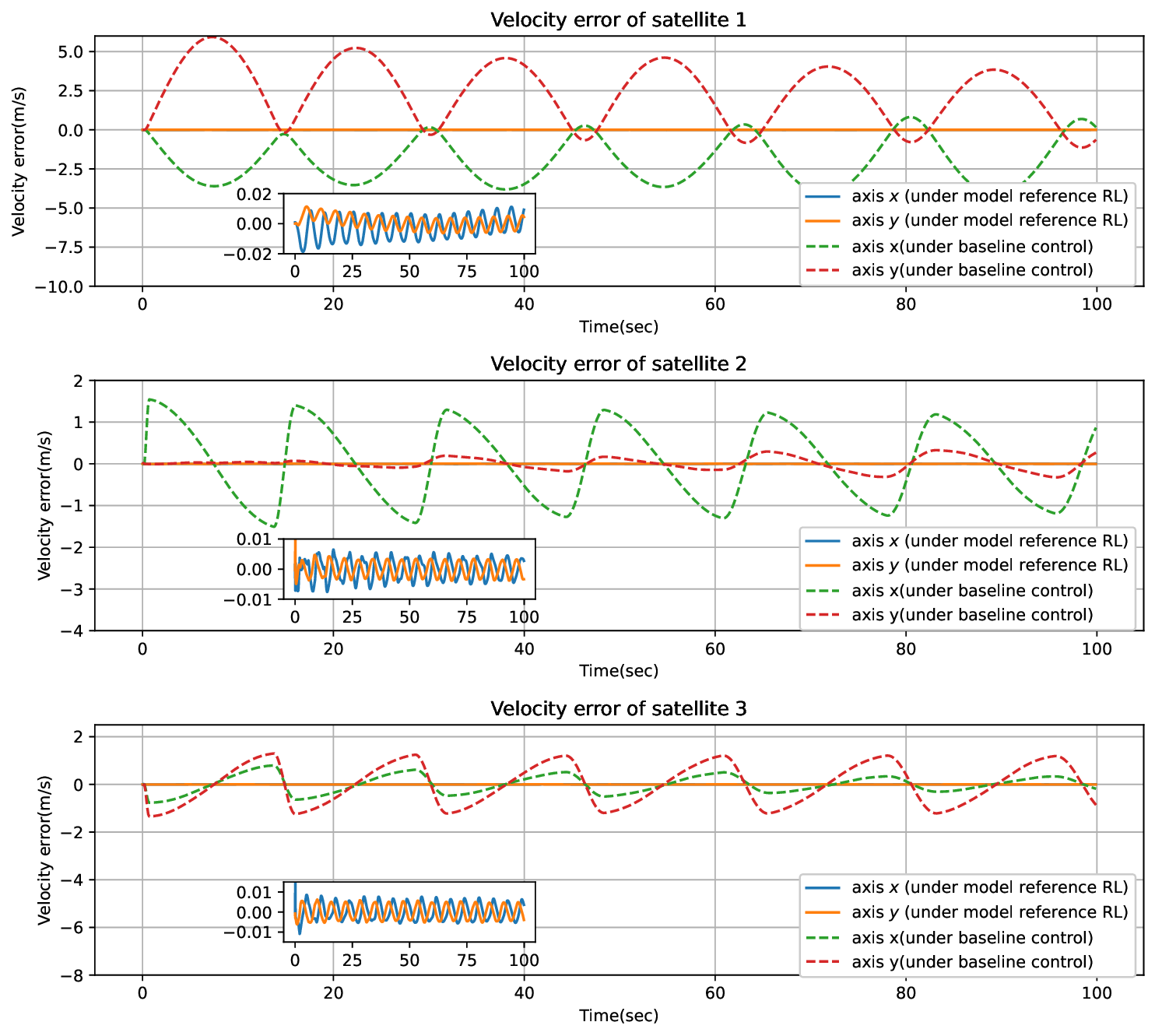}}
	\caption{Velocity error of satellites}
	\label{Fig6}
\end{figure}

\begin{figure}[!htbp] 
	\centerline{\includegraphics[width=0.48\textwidth]{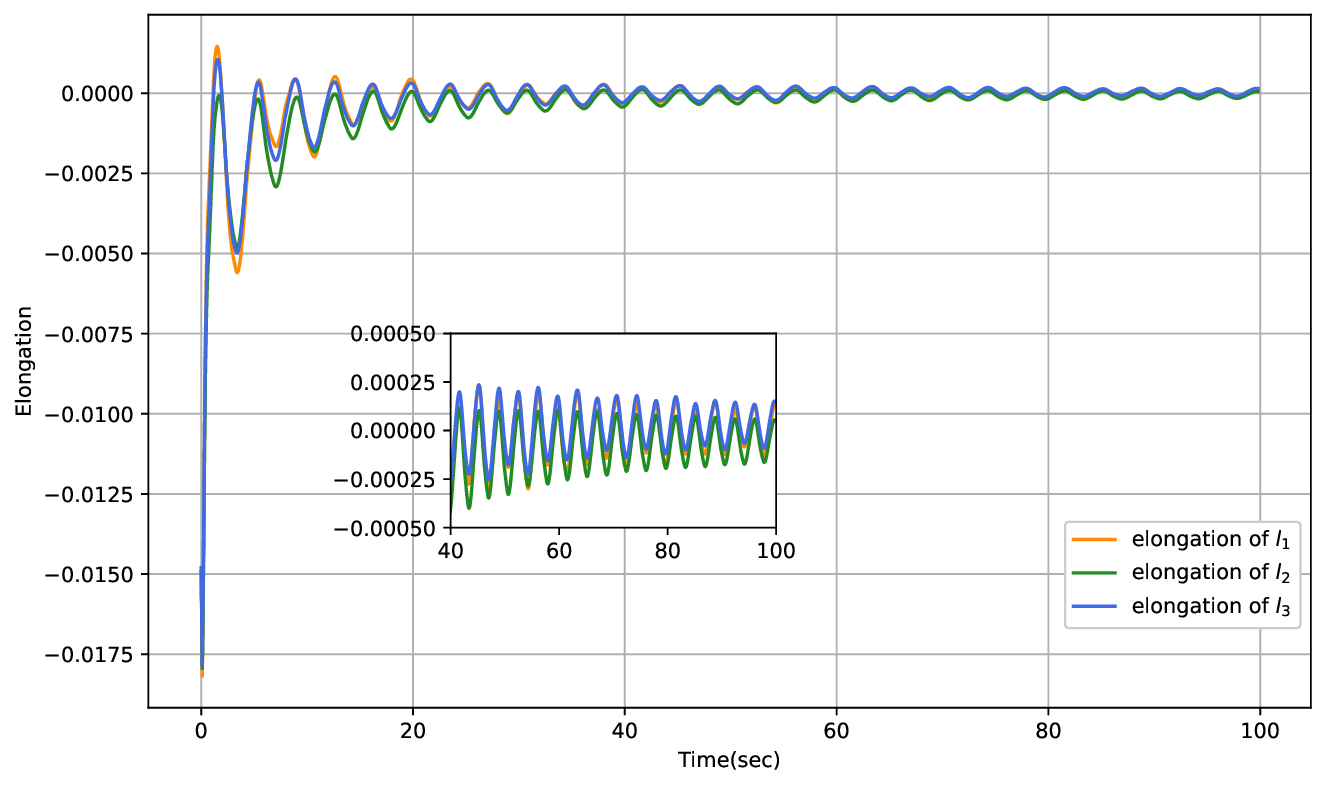}}
	\caption{Elongation of tethers}
	\label{Fig7}
\end{figure}

\begin{figure}[!htbp] 
	\centerline{\includegraphics[width=0.48\textwidth]{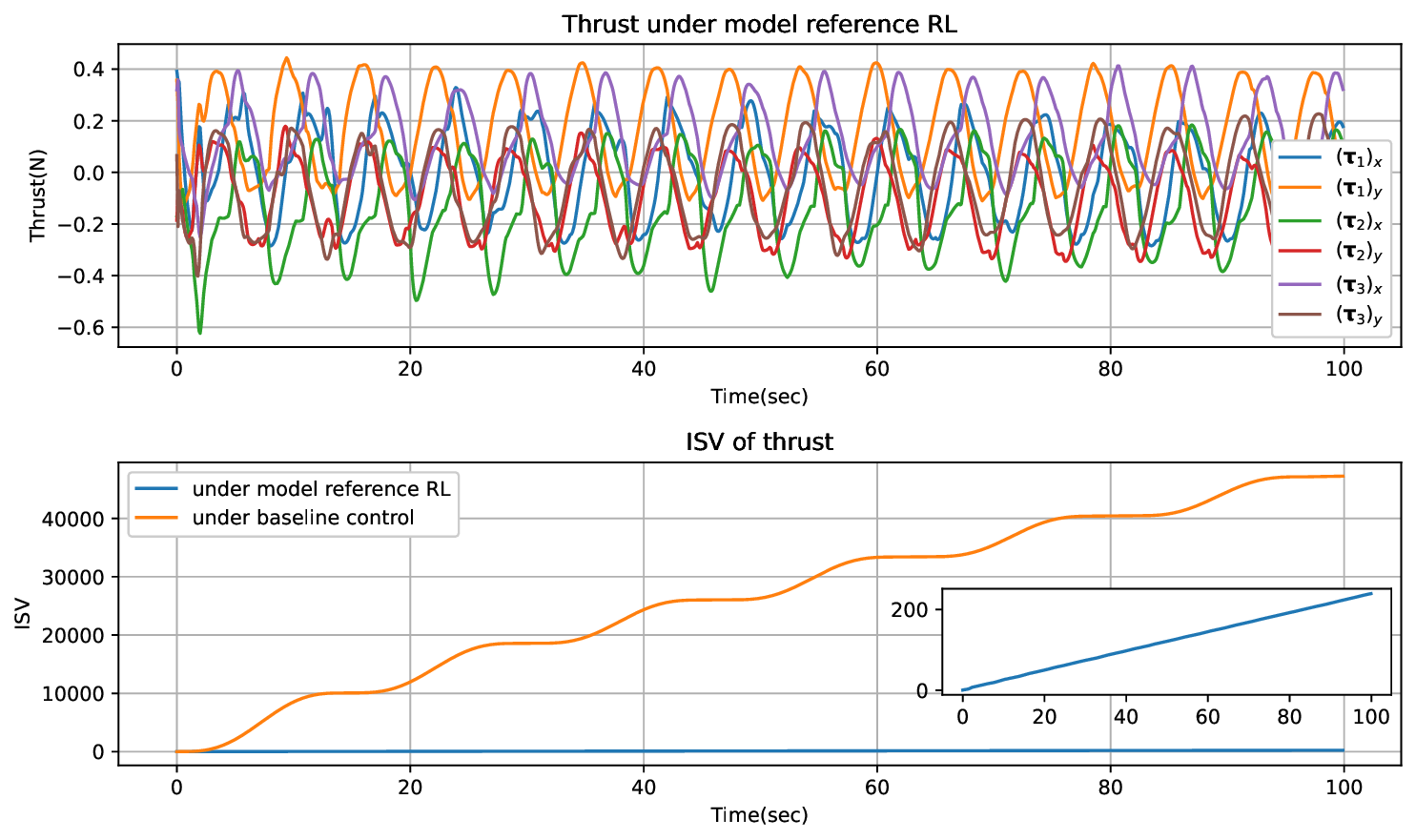}}
	\caption{Control thrust}
	\label{Fig8}
\end{figure}
Fig. \ref{Fig4}-\ref{Fig6} displays the tracking performance of our algorithom and the baseline control. For tethers, as is shown in Fig. \ref{Fig4}, the steady-state errors caused by time-varying reference trajectories are eliminated. The dynamic errors in lengths and release speed caused by outer disturbances are reduced by over 96\% and 70\% as well (from $\pm{0.8}$ to $\pm{0.025}$, $\pm{0.12}$ to $\pm{0.0.035}$ respectively). For node satellites, as illustrated in Fig. \ref{Fig5}-\ref{Fig6}, both positional and velocity dynamic errors resulting from tension and external disturbances are diminished by more than 99\% (from meter level down to centimeter level).

As shown in Figure \ref{Fig7}, elongation of tehters remain between -1.8\% and 2\textperthousand, indicating that the tether stays in a slightly tensioned or relaxed state throughout the deployment process. These simulation results align with the reward function defined in Eq. \ref{E19}. Term $R_{sat-dis,t}$ reaches its maximum value when tether lengths exactly match the distances between node satellites. It remains near maximum even when tethers in slightly taut or relaxed states. 

Control input, plotted in Fig. \ref{Fig8}, remains bounded by $\pm{1}$N for the entire deployment process. As a quantitative measure of fuel consumption, the Integral of the Squared Value (ISV) of thrust is formulated as:
\begin{equation}\label{E63}
	ISV_{u} = \int^{t_f}_{t_0}\boldsymbol{u}^T\boldsymbol{u}{\rm d}t
\end{equation}
Our control strategy cuts the ISV of thrust by two orders of magnitude compared to the baseline method.

\section{Conclusion}
\label{Sec:sec7}
In this paper, we present a model-reference reinforcement learning control strategy for space triangular tethered formation systems (TTFS). By combining a nominal model-based baseline controller with a Soft Actor–Critic compensator, the proposed method leverages the strengths of both model-based and model-free approaches. A hierarchical training framework, together with specially designed reward functions, reset conditions, and normalization procedures, effectively mitigates convergence difficulties caused by high-dimensional, strongly coupled states.

Rigorous Lyapunov-based analysis provides theoretical guarantees of closed-loop stability. Numerical simulations verify that the controller can track time-varying reference trajectories with high accuracy while maintaining the tethers in slightly tensioned or relaxed states, and drastically reducing fuel consumption compared to a conventional baseline controller.

These results confirm that the proposed control framework is an effective and scalable solution for TTFS deployment control under uncertainties and disturbances. Future work will focus on extending the approach to three-dimensional orbital motion, incorporating actuator faults, and validating the method through hardware-in-the-loop experiments.

%{\appendices
%\section*{Proof of the First Zonklar Equation}
%Appendix one text goes here.
% You can choose not to have a title for an appendix if you want by leaving the argument blank
%\section*{Proof of the Second Zonklar Equation}
%Appendix two text goes here.}

\bibliographystyle{IEEEtran}
\bibliography{IEEEabrv,abfile1} 
\vspace{5pt}

%\bf{If you include a photo:}\vspace{-33pt}
\begin{IEEEbiography}[{\includegraphics[width=1in,height=1.25in,clip,keepaspectratio]{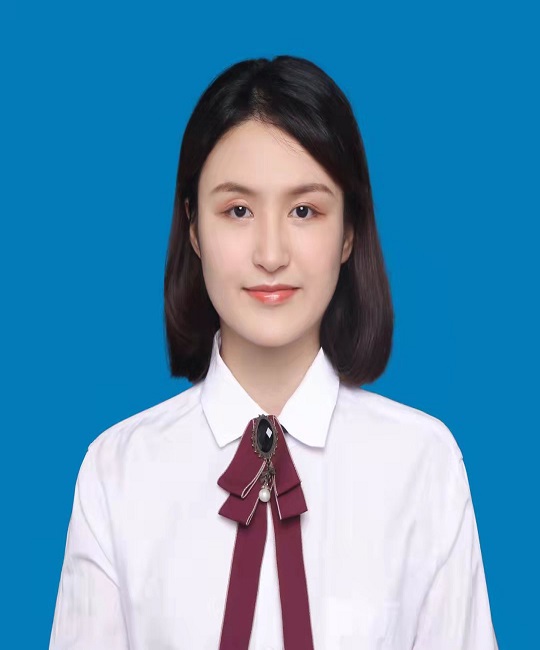}}]{Xinyi Tao}
{\space}(Student Member, IEEE) received the B.S. degree in Aircraft Control and Information Engineering from Sichuan University, Sichuan, China, in 2021. She is currently pursuing the Ph.D. degree in navigation, guidance, and control with the School of Astronautics, Northwestern Polytechnical University, Xi’an, China. 

Her current research interests include dynamics and control of tethered formation and formation control.
\end{IEEEbiography}
\vspace{-45pt}

\begin{IEEEbiography}[{\includegraphics[width=1in,height=1.25in,clip,keepaspectratio]{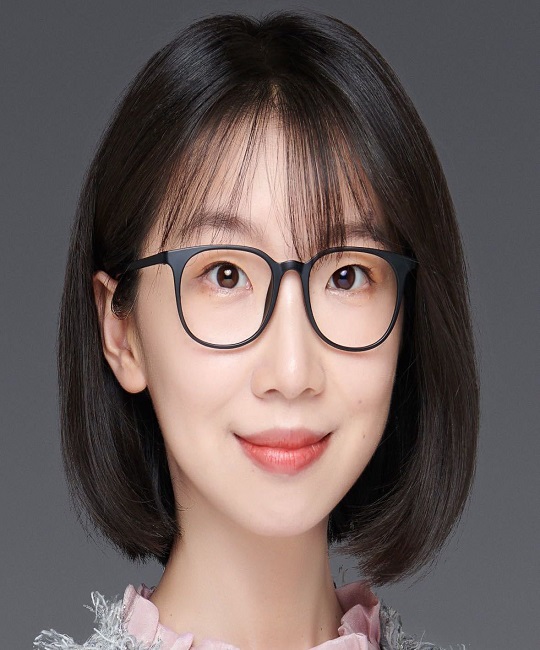}}]{Fan Zhang}
	{\space}(Member, IEEE) received the B.S.degree in detection, guidance, and control technology, in 2009, the M.S. and Ph.D. degrees in navigation, guidance, and control, form School of Astronautics, Northwestern Polytechnical University, Xi’an, China, in 2012 and 2017,respectively.
	
	She is currently a Professor with the School of Astronautics, Northwestern Polytechnical University. Her research interests include dynamic and control of tethered space robot.
\end{IEEEbiography}
\vspace{-45pt}

\begin{IEEEbiography}[{\includegraphics[width=1in,height=1.25in,clip,keepaspectratio]{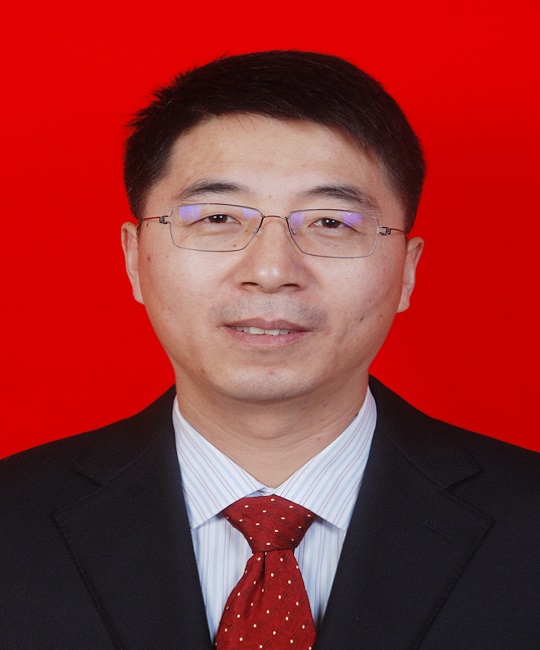}}]{Panfeng Huang}
	{\space}(Senior Member, IEEE), received the B.S. degree in test and measurement technology and the M.S. degree in navigation guidance and control from Northwestern Polytechnical University, Xi’an, China, in 1998 and 2001, respectively, and the Ph.D. degree in automation and robotics from the Chinese University of Hong Kong, Hong Kong,in 2005.
	
	He is currently a Professor with the School of Astronautics and the National Key Laboratory of Aerospace Flight Dynamics, Northwestern
	Polytechnical University, where he is the Vice Director of the Research Center for Intelligent Robotics. His current research interests include tethered space robotics, intelligent control, machine vision, and space teleoperation. 
\end{IEEEbiography}

\end{document}